\documentclass[conference]{IEEEtran}
\IEEEoverridecommandlockouts
\usepackage{graphicx}
\usepackage{booktabs}
\usepackage{subfigure}
\usepackage[skip=5pt]{caption}
\usepackage{multirow}
\usepackage{weiwAlgorithm}
\usepackage{amsmath}
\usepackage{newtxmath}
\usepackage{diagbox}
\usepackage{url}
\usepackage[colorlinks,citecolor=black,linkcolor=black,urlcolor=black]{hyperref}
\usepackage{dsfont}
\usepackage{enumitem}
\usepackage{caption}
\usepackage{balance}
\usepackage{tabularx}

\usepackage{ntheorem}

\def\BibTeX{{\rm B\kern-.05em{\sc i\kern-.025em b}\kern-.08em
    T\kern-.1667em\lower.7ex\hbox{E}\kern-.125emX}}

\DeclareMathOperator*{\argmax}{arg\,max}

\newtheorem{example}{\textit{\textbf{Example}}}
\newtheorem{theorem}{{\textit{Theorem}}}

\newtheorem{lemma}{{\textit{Lemma}}}

\newtheorem{definition}{{\textit{Definition}}}
\newtheorem*{proof}{{\textit{Proof}}\textit{.}}

\newcommand{\myparagraph}[1]{\vspace{0.2mm} \noindent \textbf{#1}.\xspace}

\newcommand{\ie}{{i.e.,}\xspace}

\newcommand{\etal}{et al.\xspace}

\newcommand{\nonl}{\renewcommand{\nl}{\let\nl\oldnl}}

\newcommand\blfootnote[1]{%
  \begingroup
  \renewcommand\thefootnote{}\footnote{#1}%
  \addtocounter{footnote}{-1}%
  \endgroup
}

\begin{document}

\title{Enhance Stability of Network by Edge Anchor}

\author{
Hongbo Qiu$^{\dag}$, Renjie Sun$^{\dag *}$, Chen Chen$^{\S *}$, Xiaoyang Wang$^{\natural}$ \vspace{1mm}\\
\fontsize{9}{9}\selectfont\itshape
$^{\dag}$\textit{Zhejiang Gongshang University}, China\vspace{1mm}
% $^{\ddag}$\textit{East China Normal University}, China\\
$^{\S}$\textit{University of Wollongong}, Australia\vspace{1mm}
$^{\natural}$\textit{The University of New South Wales}, Australia\\
\fontsize{9}{9}\selectfont\ttfamily\upshape
hongboq.zjgsu@gmail.com~renjie.sun@stu.ecnu.edu.cn\\
chenc@uow.edu.au~xiaoyang.wang1@unsw.edu.au\\

% %\vspace{8mm}
}

% \author{
% Hongbo Qiu$^{\dag}$, Renjie Sun$^{\ddag}$, Chen Chen$^{\S}$, Xiaoyang Wang$^{\natural}$, Ying Zhang$^{\flat}$ \vspace{1mm}\\
% \fontsize{9}{9}\selectfont\itshape
% $^{\dag}$\textit{Zhejiang Gongshang University}, China
% $^{\ddag}$\textit{East China Normal University}, China\\
% $^{\S}$\textit{University of Wollongong}, Australia
% $^{\natural}$\textit{The University of New South Wales}, Australia\vspace{1mm}
% $^{\flat}$\textit{University of Technology Sydney}, Australia\\
% \fontsize{9}{9}\selectfont\ttfamily\upshape
% hongboq.zjgsu@gmail.com~renjie.sun@stu.ecnu.edu.cn\\
% chenc@uow.edu.au~xiaoyang.wang1@unsw.edu.au~ying.zhang@uts.edu.au\\

% % %\vspace{8mm}
% }

\maketitle

\begin{abstract}

With the rapid growth of online social networks, strengthening their stability has emerged as a key research focus.
This study aims to identify influential relationships that significantly impact community stability.
In this paper, we introduce and explore the anchor trussness reinforcement problem to reinforce the overall user engagement of networks by anchoring some edges. 
Specifically, for a given graph $G$ and a budget $b$, we aim to identify $b$ edges whose anchoring maximizes the trussness gain, which is the cumulative increment of trussness across all edges in $G$. 
We establish the NP-hardness of the problem.
To address this problem, we introduce a greedy framework that iteratively selects the current best edge. 
To scale for larger networks, we first propose an upward-route method to constrain potential trussness increment edges. 
Augmented with a support check strategy, this approach enables the efficient computation of the trussness gain for anchoring one edge. 
Then, we design a classification tree structure to minimize redundant computations in each iteration by organizing edges based on their trussness.
We conduct extensive experiments on 8 real-world networks to validate the efficiency and effectiveness of the proposed model and methods.
\blfootnote{*Corresponding author}
\end{abstract}

\section{Introduction}
\label{sec:intro}
% movitated example 需要体现anchor边的合理性 和anchor点的进行比较
% challenge 和 contribution 写详细

%
%介绍图，以及本文做的关键边检测
Graphs are a powerful and widely used tool for analyzing social networks, as they can represent relationships between different entities.
Recently, there has been increasing interest in understanding user engagement \cite{zhang2017olak,zhangfanefficiently2018,QingyuanGobal2020} and examining the stability and cohesiveness of social networks \cite{zhangfanQuantifyingNodeImportance,BuOnImprovingtheCohesiveness}. 
Empirical studies have demonstrated that user participation or departure significantly affects social networks, with these changes often being influenced by the behaviors of their connections \cite{Levsocial}.
%Empirical studies have shown that user participation or departure has a profound impact on social networks, often influenced by the behaviors of their connections \cite{Levsocial}. 
For instance, when an active user leaves a network, it may trigger a cascade effect that reduces overall user engagement and weakens relationships within the network \cite{QingyuanGobal2020,Seki2017TheMO}. 
Key users and relationships play a critical role in fostering engagement, promoting information dissemination, and strengthening cooperation within networks \cite{wanginfluencemax,wang2025effective}.
% For example,  can suppress the spread of misinformation by blocking key users and their relationships.

%介绍k-truss，提一下truss的优点
In graph theory, a $k$-truss is a dense subgraph where every edge must be included in at least $k-2$ triangles (\ie support). 
Given a graph, the $k$-truss can be calculated by iteratively removing edges with support less than $k-2$.
Each edge has a trussness value, indicating the largest $k$ for which a $k$-truss containing the edge exists.
%representing the largest $k$ such that there exists a $k$-truss containing the edge.
The $k$-truss has many properties, such as higher density, strong connectivity, and polynomial-time computations.
Thus it is widely utilized for discovering cohesive communities \cite{liuqingtrussbesed,zhaojunCommunity,VLDB-eff-community-search,huang2014querying,Esratruss2017,wu2020maximum,sun2022diversified}.

\begin{figure}
    \centering
    \includegraphics[width=0.7\linewidth]{ 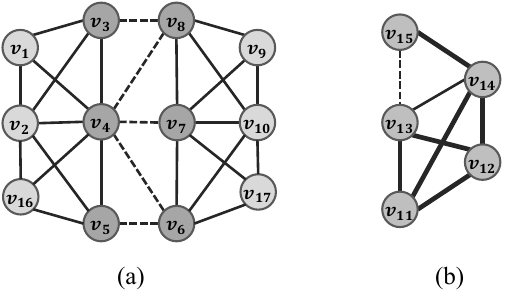}
    %\vspace{-1mm}
    \caption{A toy example}
    \label{fig:te}
    %\vspace{-1mm}
\end{figure}

\myparagraph{Motivation}
The size of a $k$-truss serves as a feasible indicator of network stability.
%As the size of $k$-truss is a feasible indicator of network stability, 
Thus Zhang \etal \cite{zhangfanefficiently2018} study the anchor $k$-truss (AKT) problem, which focuses on expanding a specific $k$-truss subgraph by selecting $b$ anchor vertices to maximize the size of the $k$-truss. In AKT problem, anchor vertices are assumed to always remain in the $k$-truss, regardless of their original connectivity. For example, in a social network, incentivizing key users to stay active can encourage continued engagement and support other users' participation \cite{bhawalkar2015preventing}.

However, the AKT problem is limited to local enhancements of the network. Specifically, it only expands the $k$-truss with a particular $k$ value. As described in \cite{zhangfanefficiently2018}, anchoring a vertex can only increase the trussness of edges with trussness equal to $k-1$, and by at most 1. Besides, the valid anchor vertices are restricted to endpoints of edges with trussness equal to $k-1$, meaning that other vertices cannot contribute to the expansion. This limitation prevents a more comprehensive enhancement of the network structure. Furthermore, the $k$-truss model inherently evaluates network stability based on edge strength \cite{chenLocating,DBLP:conf/ijcai/ZhuZCWZL19}, \ie the support of edges, rather than vertex persistence. However, the AKT problem focuses on anchoring vertices rather than strengthening edges, which contradicts the fundamental principles of the truss model. In numerous real-world applications, interactions or relationships (edges) play a more crucial role than individual entities (vertices) \cite{Song_2023,ouyangQuantifying,alexandre2024critical}. Network stability is often better preserved by maintaining high-quality connections, rather than simply ensuring the existence of certain vertices.

Based on the above analysis, we introduce and investigate the anchor trussness reinforcement (ATR) problem, which aims to select $b$ anchor edges to maximize the overall trussness gain across the entire graph. Compared to the AKT problem, our problem focuses on enhancing the trussness of all edges, thereby improving the global structural cohesion and robustness of the network rather than just expanding a local $k$-truss. In ATR problem, anchoring an edge means that it remains persistently in any truss structure and continuously provides support to all edges forming triangles with it, regardless of structural changes. Since the $k$-truss model is defined based on edge support (\ie the number of triangles an edge belongs to), we set the support of anchor edges to \textit{infinity} as a computational abstraction. This ensures that they always contribute to triangle formation under any conditions, reinforcing the stability of the entire network.

\myparagraph{Application}
%The ATR problem can find many practical applications, and we list some of them as follows.
The ATR problem has various real-world applications, some of which are outlined below.

\begin{itemize}[leftmargin=1em]
\item \textit{Enhancing the stability of social networks}. Maintaining overall stability in social networks is essential for preserving engagement, sustaining information flow, and ensuring community integrity. Relationships between users form the foundation of social interactions, and the loss of key connections can lead to community fragmentation, reduced participation, and weakened information diffusion. Traditional approaches to network stability often focus on identifying and retaining influential users, but this strategy overlooks the structural importance of critical relationships that maintain network cohesion. By anchoring certain key social connections, the overall stability of the network can be significantly enhanced. Reinforcing essential interpersonal ties ensures that communities remain connected, even if some users become inactive or reduce their interactions. This approach helps prevent network fragmentation while ensuring that important social structures remain intact, enabling continuous engagement and interaction. 
\item \textit{Enhancing the stability of transportation networks}. Maintaining overall stability and resilience in transportation networks is essential for ensuring efficient and reliable mobility. Transportation systems frequently experience disruptions due to traffic congestion, accidents, or other external factors. When key connections in the network are weakened or lost, they can trigger cascading failures, leading to severe delays and inefficiencies across the system. To address this issue, the ATR problem studied in this paper can be leveraged to identify critical connections in the transportation network. By detecting and reinforcing these essential links, transportation networks can achieve greater adaptability, reduced vulnerability to disruptions, and improved operational efficiency.
\end{itemize}

 %\cite{FORTUNATO201075}. 
%In the field of information dissemination, selectively reinforcing critical connections in social networks significantly boosts the efficiency and reach of information spread, ensuring that important messages propagate quickly throughout the network \cite{weicheninflumax}. 
%Anchoring edges in network security enhance the resilience of critical connections by preventing their removal, thus strengthening the graph's overall structure against attacks\cite{zhangFinding}.

%我想表达用anchor边来做我们的问题能更精准地检测关键边，anchor点只能检测关键点，他们的方法不适用
\begin{example}
As depicted in Fig. \ref{fig:te}, we first consider the vertex anchoring approach \cite{zhangfanefficiently2018} with  $k=4$.
In Fig. \ref{fig:te}(a), the trussness of solid edges is 4, while that of dotted edges is 3.
Anchoring vertex $v_8$  ensures that edges $(v_3,v_8)$ and $(v_4,v_8)$ remain in the 4-truss, as they form the triangle $\Delta_{v_3v_4v_8}$. 
Anchoring vertex $v_6$  ensures that edges $(v_4,v_6)$ and $(v_5,v_6)$ remain in the 4-truss, as they form the triangle $\Delta_{v_4v_5v_6}$. 
This has the same effect as directly anchoring edge $(v_3,v_8)$ and $(v_5,v_6)$. 
Notably, anchoring $(v_3,v_8)$ and $(v_5,v_6)$ also increases the trussness of 3-truss edges. 
In Fig. \ref{fig:te}(b), each bolded edge is assumed to belong to a separate clique of size 5, ensuring a trussness of 5.
Here, the trussness of edge $(v_{13},v_{15})$ is 3, while  edge $(v_{13},v_{14})$ has a trussness of 4.
Anchoring vertex $v_{14}$ in Fig. \ref{fig:te}(b) has no effect when $k=4$, because $(v_{13},v_{14})$ is already a 4-truss edge.
However, directly anchoring edge $(v_{13},v_{15})$ increases the trussness of $(v_{13},v_{14})$ to 5.
% Note that the anchoring of $(v_{13},v_{15})$ leads to trussness increase of $(v_{13},v_{14})$.
This example highlights that the edge anchoring approach can consider trussness increase globally, whereas choosing a suitable $k$ is challenging for the vertex anchoring approach \cite{zhangfanefficiently2018}.
Furthermore, it demonstrates that directly anchoring edges more effectively target critical edges, especially when a vertex has many incident edges.
\end{example}

\myparagraph{Challenges} 
%\textcolor{red}{As far as we know, this is the first study to explore the ATR problem and establish its NP-hardness.}
To the best of our knowledge, we are the first to study the ATR problem.
We prove that the problem is NP-hard.
While truss decomposition can be completed in polynomial time \cite{jiawangTrussdecomposition}, an exact solution requires exhaustively evaluating the trussness gain for every possible combination of $b$ anchor edges, which is computationally infeasible.
Moreover, we prove that the trussness gain function is non-submodular, further complicating the problem.
Although estimating the global trussness gain for multiple anchors is impractical, we observe that trussness changes are relatively localized when anchoring a single edge.
This observation leads us to use a greedy heuristic to iteratively choose the optimal anchor edges within a given budget $b$.
%This observation motivates us to adopt a greedy heuristic to iteratively select the best anchor edge under a given budget $b$.
However, even with a greedy approach, the direct implementation is prohibitively time-consuming.
This is because each edge in the graph is a potential anchor, resulting in a large candidate set that must be evaluated for its trussness gain. 
Besides, after selecting an anchor edge in each iteration, the trussness of other edges in the graph may change, necessitating the re-computation of trussness gain for all edges in subsequent iterations.
These challenges are further exacerbated as the graph size grows.

Existing studies on related problems, such as the anchor $k$-core and anchor $k$-truss problems, provide limited solutions for our problem.
Bhawalkar \etal \cite{kshpreventing2015} introduced the anchor $k$-core problem, which was further explored by Zhang \etal \cite{zhang2017olak} and Linghu \etal \cite{QingyuanGobal2020}. However, these approaches rely on vertex deletion orders, which are not applicable to the $k$-truss as it is defined on edges and triangles.
While the $k$-core treats all edges equally, the $k$-truss evaluates edge strength based on the number of triangles they form, providing a more nuanced model of network structure.
Zhang \etal \cite{zhangfanefficiently2018} introduced an efficient algorithm for the anchor $k$-truss problem, focusing on selecting $b$ vertices as anchors to ensure that more vertices can be retained within a specific $k$-truss structure. The selected anchor vertices are the endpoints of the edges with trussness equal to $k-1$. In contrast, our problem focuses on increasing the overall trussness across the entire graph rather than targeting a particular $k$-truss. As a result, the anchoring edges identified by our problem are distributed across different trussness levels, rather than being restricted to edges within a single $k$-truss. Due to this fundamental difference, the method in \cite{zhangfanefficiently2018} is not a viable approach for our problem.
% Zhang \etal \cite{zhangfanefficiently2018} proposed an efficient algorithm for the anchor $k$-truss problem, focusing on selecting $b$ vertices as anchors to enhance the $k$-truss structure.
% In contrast, our problem requires selecting anchor edges to improve global trussness, significantly expanding the search space as every edge in the graph is a potential anchor.
% Additionally, prior work only considers $k$-trusses locally, \ie a specific $k$ value, whereas our problem focuses on global trussness improvement.
Consequently, our problem presents unique challenges that necessitate the development of advanced strategies to accelerate or avoid the computation of trussness gain for each candidate anchor.

\myparagraph{Our solution}
Given the computational challenges of the ATR problem, we employ a greedy heuristic to iteratively select the optimal anchor edge. 
In each iteration, the trussness gain of each edge is computed, and the edge with the largest gain is chosen. 
Then a straightforward approach is to utilize truss decomposition to calculate the updated trussness for each edge.
However, recomputing the trussness of each edge after every iteration by using this method is costly and impractical for large graphs.
%However, recomputing the trussness of each edge after each iteration using this method is very expensive and impractical for large graphs.

To address this issue, we draw inspiration from \cite{zhangfanefficiently2018} and revisit the deletion order in truss decomposition, leveraging it to accelerate our algorithm.
Specifically, when an edge is anchored, it prevents certain edges from being removed during the truss decomposition. We refer to these edges as followers.
We introduce the concept of the \textit{upward-route} and prove that only edges along this route can become followers of the anchor edge.
By focusing on the upward-route and applying a support check mechanism, we can efficiently identify the followers of the anchor.
Furthermore, to reuse intermediate results from previous iterations, we propose a tree structure that groups edges into manageable tree nodes.
This structure enables efficient determination of whether previously computed results for a candidate anchor can be reused, thereby avoiding redundant computations.
If a tree node cannot be utilized again, the follower computation is conducted solely within that node.
%If a tree node cannot be reused, follower computation is \textcolor{red}{carried out} locally within that tree node.
By combining these techniques, we develop the \texttt{GAS} algorithm, which efficiently identifies the best anchor edge in each iteration, providing a practical solution to the ATR problem.

% Considering the hardness of our problem, we adopt the greedy approach. 
% In each iteration, we check trussness gain that each edge brings and choose the best one. 
% A straightforward way is to use truss decomposition on whole graph. Unfortunately, such operation is still very time-cost. 
% Motivated by \cite{zhangfanefficiently2018}, we revisit the deletion order in truss decomposition, and find it can be used to accelerate our algorithm. 
% When an edge is anchored, the deletion after it may possible increase trussness by 1, named followers. 
% Thus we propose the upward route and prove only edges along the upward route can become the followers of the anchor. 
% Then we focus on the upward route and use support check to find the exact followers.
% Considering the results with last iteration, we also propose an structure to classify edges, which can help us to avoid re-computation of reusable results. 
% Combined with above techniques, our final algorithm is formed, named as GAS, which can efficiently solve the problem.

\myparagraph{Contributions} Our main contributions are as follows.
\begin{itemize}[leftmargin=1em]
    \item  We introduce the anchor trussness reinforcement problem, which aims to select $b$ edges as anchors to maximize the global trussness gain and enhance network stability. We formally define the problem and prove its NP-hardness.
    
    \item We revisit the edge deletion order in truss decomposition and partition edges into layers. Based on the orders of edges' deletion, we propose the concept of an upward-route rooted at the anchor edge, which significantly narrows the search space. Combined with a support check process, this approach enables efficient computation of trussness gain when selecting an edge as an anchor in social networks.
    
    \item  We develop a tree structure to categorize edges based on their triangle connectivity and trussness values. This structure allows us to precisely identify reusable results for follower edges after anchoring an edge in each round, thereby avoiding extensive recomputation.
    
    \item We perform extensive experiments on 8 real-world datasets to assess the effectiveness and efficiency of the proposed techniques.
    %We conduct comprehensive experiments on 8 real-world datasets to evaluate the effectiveness and efficiency of the proposed techniques. 
\end{itemize}

% \textbf{Roadmap}. The rest of this paper is organized as follows. Section \ref{sec:pre} formulates the problem definition. Section \ref{sec:skc} introduces the baseline algorithm along with our proposed techniques. Experimental results are presented in Section \ref{sec:exp}. Related work is discussed in Section \ref{sec:rel}, and we conclude the paper in Section \ref{sec:conc}.

\begin{table}[h]
    \centering
    \footnotesize
    \caption{Summary of notations}
    \renewcommand\arraystretch{1.2}
    \label{summary of notations}
    \begin{tabular}{|c|l|}
        \hline
        \textbf{Notation}   &\ \textbf{Definition} \\ 
        \hline\hline
        $G$                 &\ An unweighted undirected graph \\ 
        \hline
        $G_{e_x}/G_A$          &\ The graph G after anchor edge $e_x$ $/$ anchor set $A$ \\ 
        \hline    
        $E,V$         &\ The edge set and the vertex set of graph $G$ \\ \hline
        $S$                 &\ The subgraph of $G$ \\ \hline
        $N(u,S)$            &\ The set of neighbor of vertex $u$ in $S$ \\ \hline
        $sup(e,S)$          &\ The number of triangle that containing $e$ in $S$ \\ \hline
        $T_k(G)$            &\ The $k$-truss of $G$ \\ \hline
        $t(e)$              &\ The trussness of edge $e$ \\ \hline
        $t^A(e)$              &\ The trussness of edge $e$ after anchor $A$\\ \hline
        $k$                 &\ The support constraint \\ \hline
        $b$                 &\ The number of edge budget \\ \hline
        $A$                 &\ The anchored edge set \\ \hline
        $TG(A,G)$  &\ The trussness gain after anchor edge set $A$ in $G$ \\ \hline
        %$N_{\triangle}(e,G)$          &\ The triangle neighbors set of an edge $e$ in graph $G$ \\ \hline
        $\triangle_{uvw}$ &\ The triangle that including three vertices $u,v,w$ \\ \hline
        $F(e,G)$ &\ The followers of the anchor $e$ in $G$ \\ \hline
    \end{tabular}
\end{table}
\section{Preliminaries}
\label{sec:pre}

%\textcolor{red}{In this section, we first present some related concepts, followed by a formal definition of the problem and an analysis of its computational hardness. The mathematical notations commonly used throughout the paper are summarized in Table~\ref{summary of notations}.}
In this section, we first introduce some related concepts, then formally define the problem and establish its computational hardness. Frequently used mathematical notations throughout the paper are summarized in Table~\ref{summary of notations}.

\subsection{Problem definition}  

We consider an unweighted and undirected graph $G = (V, E)$, where $V$ and $E$ represent the sets of vertices and edges, respectively. 
Let $n = | V |$ and $m = | E |$ denote the number of vertices and edges in $G$, respectively. 
For a given subgraph $S$ of $G$, we use $N(u, S)$ to denote the neighbor set of $u$ in $S$, and $deg(u, S) = |N(u, S)|$ to specify its degree. A triangle, denoted by $\Delta_{uvw}$, consists of three mutually connected vertices $u$, $v$ and $w$.
% Given an edge e, a containin-e-triangle $\triangle_e$ is a triangle that contains e. 
% Let $\triangle_e$ denote a triangle that contains edge $e$.
%A $\triangle_e$ is a triangle that contains edge $e$.
% Before introducing the problem definition, we first give some concepts.

\begin{definition}[Support]
Given a subgraph $S$ of $G$, the support of an edge $e(u,v)$ in $S$, denoted by $sup(e,S)$, is the number of triangles in $S$ that containing $e$, \ie $sup(e,S)=|N(u,S) \cap N(v,S)|$.
\end{definition}
%$\ie sup(e, S) = |\{\Delta_{uvw}:\Delta_{uvw} \in \Delta_S\}|$

\begin{definition}[$k$-truss]
Given a graph $G$, a subgraph $S$ is the $k$-truss of $G$, denoted by $T_k(G)$, if $(i)$ $sup(e,S) \geq k-2$ for each edge $e$ in $S$; $(ii)$ $S$ is maximal, \ie any supergraph of $S$ does not satisfy condition $(i)$; and $(iii)$ there is no isolated vertices in $S$.
\end{definition}

\begin{definition}[Trussness]
Given a graph $G$, the trussness of an edge $e$ in $G$, denoted by $t(e)$, is the largest $k$ such that there exists a $k$-truss containing $e$, \ie $t(e) = \max\{k|e \in T_k(G)\}$.
\end{definition}

To compute the trussness for each edge $e \in G$, we utilize the truss decomposition method \cite{jiawangTrussdecomposition}, whose details are shown in Algorithm~\ref{alg:truss-decomp}.
For each $k$ starting from 2, the algorithm iteratively removes the edges with support no larger than $k-2$.
We set the trussness of the removed edge as $t(e) = k$.
When $e$ is removed, the support of other edges forming a triangle with $e$ is reduced by 1.
%Upon the removal of $e$, the support of other edges that form a triangle with $e$ will be decreased by 1.
This process continues until the support of all the remaining edges is larger than $k-2$.
The time complexity is $O(m^{1.5})$ \cite{jiawangTrussdecomposition}.

% \begin{definition} [truss decomposition]
% Given a graph $G$, truss decomposition of $G$ is to compute the trussness of each edge in $E(G)$.
% \end{definition}

In this paper, when an edge $e$ in $G$ is deemed ``anchored'', its support is considered to be positive infinity, \ie $sup(e, G) = +\infty$ for each anchored edge $e$. Each anchored edge is termed an ``anchor'' or ``anchor edge''. The collection of all anchor edges is represented by $A$.
The presence of anchor edges can increase the trussness of other edges.
We use $G_A$ to represent the graph $G$ with the anchor edge set $A$, and $t^A(e)$ to denote the trussness of edge $e$ in $G_A$.
%After edge $e$ is anchored, those non-anchored edges whose trussness increases are called \textit{follower} of $e$, denoted by $F(e,G)$, \ie $F(e,G) = \{e' \in G | t^{\{e\}}(e') > t(e')\}$. 
Note that, the computation of truss decomposition on $G_A$ is essentially the same as on $G$, except that the anchor edges are always retained in $G_A$.

\begin{algorithm}[t]
    \SetVline
    \footnotesize
    \caption{{\textbf{TrussDecomp}$(G)$}}
    \label{alg:truss-decomp}
    \Input{$G$ : the graph}		
    \Output{the trussness $t(e)$ of each edge $e \in G$}
    
    \State{$k \leftarrow 2$}   
    \While{exist edge in $G$}
    {
        \While{$\exists e(u,v) \in G$ with $sup(e,G) \leq k-2$}
        {
            \For{$w \in N(u,G) \cap N(v,G)$)}
            {
                \State{$sup((u,w),G) \gets sup((u,w),G) - 1$}
                \State{$sup((v,w),G) \gets sup((v,w),G) - 1$}
            }
            \State{$t(e) \gets k$}
            \State{remove $e$ from $G$}
        }
        \State{$k \gets k+1$}
    }
    \Return{$t(e)$ for each edge $e \in G$};
\end{algorithm}

Existing research on $k$-truss maximization problem predominantly focuses on anchor vertices and specific $k$ values \cite{zhangfanefficiently2018}. 
However, it is more pertinent to emphasize the strength of the connections between pairs of vertices, \ie the strength of the edges.
Moreover, enhancing the cohesion of the overall community proves more advantageous than reinforcing the $k$-truss for a fixed $k$ value.
In many studies~\cite{VLDB-eff-community-search,alifirmtruss2022,huang2014querying,Esratruss2017}, trussness has been used to evaluate community cohesion, where a higher level of trussness among edges within a community indicates a more stable and well-connected social structure.
Thus, in this paper, we first give the Definition \ref{def:tg}, then propose and investigate the anchor trussness reinforcement (ATR) problem.

\begin{definition}[Trussness gain]\label{def:tg}
    Given a graph $G=(V,E)$ and an anchored edge set $A$, the trussness gain of $G$ after anchoring $A$, denoted by $TG(A,G)$, is the total enhancement in trussness for all edges in $E \backslash A$, \ie $TG(A,G)= \sum_{e \in E \backslash A}(t^A(e) - t(e))$. 
\end{definition}
 
\myparagraph{Problem statement}
Given a graph $G$ and a budget $b$, the ATR problem aims to identify an edge set $A$ of $b$ edges in $G$ such that $TG(A,G)$ is maximized.

\subsection{Problem analyse}

\begin{theorem}
Given a graph $G$, the ATR problem is NP-hard.
\end{theorem}

\begin{proof}
    We reduce the maximum coverage problem \cite{karpMCproblem} to the ATR problem.
    Given a budget $b$ and a collection of sets, each containing some elements, the maximum coverage problem seeks to select $b$ sets that cover the most elements.
    %Given a budget $b$ and a collection of sets where each set contains some elements, the maximum coverage problem aims to find $b$ sets to cover the largest number of elements. 
    We consider an arbitrary instance of maximum coverage problem with $s$ sets $\{ T_1,T_2,\dots,T_s \}$ and $t$ elements $\{e_1,e_2,\dots,e_t\} = \cup_{1\leq i \leq s}T_i$. 
    We assume that $b < s < t$. 
    %Then we construct a corresponding instance of the ATR problem in a graph $G$ as follows.
    Following this, we proceed to construct an instance of the ATR problem within the graph $G$ as outlined below.

    %\textcolor{blue}{
    Fig. \ref{fig:np-hard}(a) is a constructed example for $s=3, t=4$.
    We divide $G$ into three parts, $E_a, E_f$ and some $(t+3)$-cliques (fully connected subgraph formed by $t+3$ vertices). The $E_a$ part contains $s$ edges, $\ie$ $E_a = \{a_1, a_2, \dots, a_s\}$. 
    Each edge $a_i$ corresponds to $T_i$ in the maximum coverage problem instance. 
    The $E_f$ part contains $t$ edges, \ie  $E_f = \{f_1,f_2,\dots,f_t\}$. 
    Each edge $f_j$ corresponds to $e_j$ in the maximum coverage problem instance. 
    For each edge $a_i \in E_a$, if its corresponding $T_i$ contains $e_j$, we add a $(t+3)$-clique (the trussness of each edge in $(t+3)$-clique is $t+3$). We then make $a_i$, $f_j$ and an arbitrarily chosen edge of that $(t+3)$-clique form a triangle (as shown in the left of Fig. \ref{fig:np-hard}(b)).
    For each $f_j \in E_f$, we add $2t$ $(t+3)$-cliques.
    Then, we create $t$ triangles for $f_j$, each formed by $f_j$ and any two edges of two $(t+3)$-cliques (as shown in the right of Fig. \ref{fig:np-hard}(b)).
    At this point, the construction is finished.
    %Then the construction is completed.
%    }

   % \textcolor{blue}{
    With the construction, we can have the following results: $(i)$ The trussness of edge $a_i \in E_a$ is $|T_i|+2 \leq t+2$. $(ii)$ The trussness of edge $f_j \in E_f$ is $t+2$, because we made it form $t$ triangles with the edge whose trussness is $t+3$ during the construction process. $(iii)$ Anchoring any edge $a_i \in E_a$ can only increase the trussness of the edge in $E_f$ that forms a triangle with it by 1. $(iv)$ Even if multiple edges in $E_a$ are anchored, the trussness of the edges in $E_f$ can only be increased by 1. $(v)$ Anchoring any edge in $G \backslash E_a$ cannot obtain trussness gain.
    %}  
    %$sup(a_i,G) \leq t\textless sup(f_i,G)$. 
    %According to the truss decomposition of $G$, which is deleting the edges with support less than $k-2$, where $k$ is from $2$ to $k_{max}$ iteratively. 
   % Then we have the trussness of the deleted edges in the current iteration is $k-1$ since we can only delete $a_i$ when $k = sup(a_i,G)+3$, the trussness of each $a_i$ in $E_i$ is $sup(a_i,G)+2$. 
    %Similarly, we can delete $f_i$ when $k = t+1$ as the deletion of each $a_i$ in $E_a$, but $f_i$ is not deleted when $k = t$ as the $(t+1)$-clique. 
    %Thus, the trussness of each $f_i$ in $E_f$ is $t$. Likewise, the trussness of each edge in the $(t+1)$-clique is $t+1$.
    %We can guarantee that, the trussness of $a_i$ remains the same, even if we anchor all edges except $a_i$ of triangles that contained $a_i$, since $a_i$ only can be deleted when $k = sup(a_i,G)+3$. 
   % Since the edges in $(t+1)$-cliques will be deleted when $k = t+2$, we anchor multiple edges cannot increase the trussness of $f_i$ in $E_f$ to $t+2$, due to $b < s < t$. 
   % Thus, we know that the trussness of $f_i$ increases by 1 when we anchor at least one neighbor of $f_i$ in $E_i$. 
  % \textcolor{blue}{
    By doing this, we ensure that only the edges in $E_a$ can obtain trussness gain, and the trussness gain is the number of edges that form triangles with it in $E_f$. 
    Therefore, the optimal solution to the ATR problem is equivalent to that of the MC problem. Given that the maximum coverage problem is NP-hard, it follows that the ATR problem is also NP-hard for any $b$.
    %Consequently, the optimal solution of the ATR corresponds to the optimal solution of the MC problem. 
    %Since the maximum coverage problem is NP-hard, we \textcolor{red}{demonstrate} that the ATR problem is NP-hard for any $b$.
   % }
    
\end{proof}

\begin{figure}
	\centering
     \subfigure[Construction example] {\includegraphics[width=1\linewidth]{ 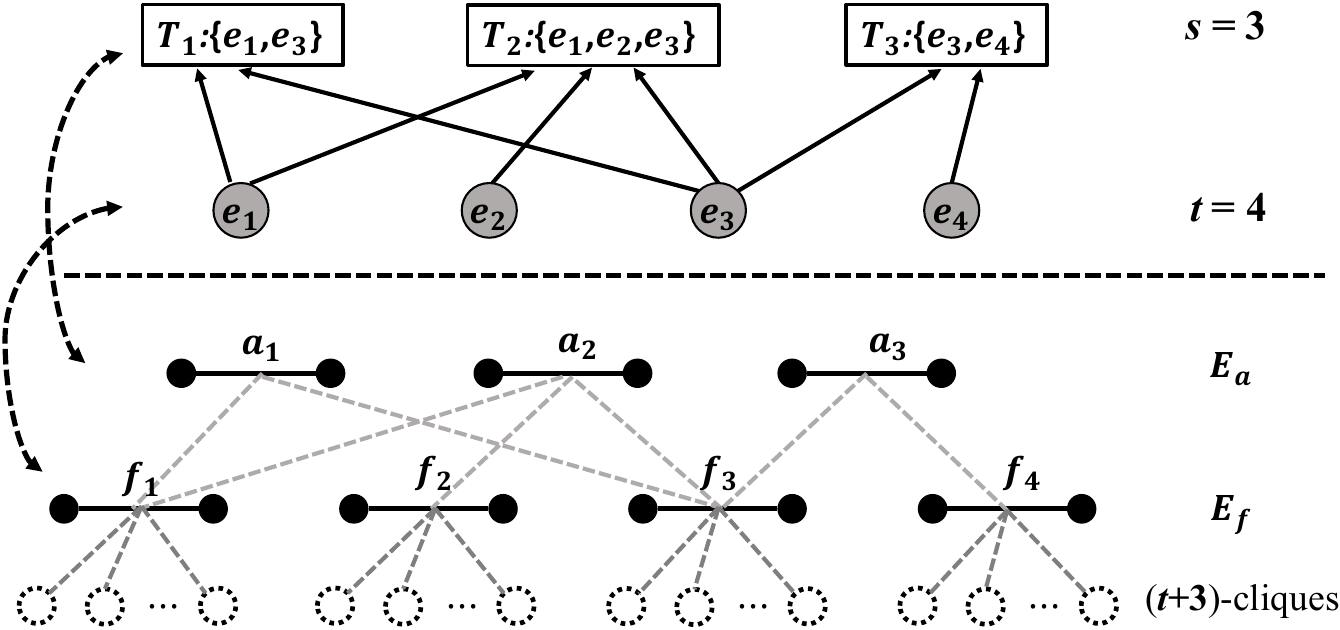}}
     
    \subfigure[Structure illustration] {\includegraphics[width=1\linewidth]{ 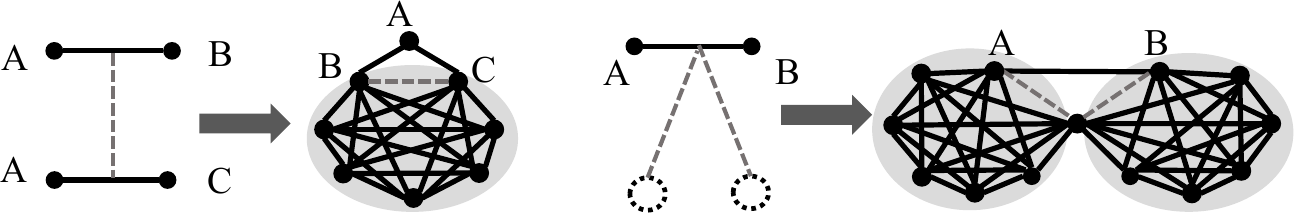}}
 
	% \vspace{-2mm}
	\caption{{Example of NP-hard proof}}
	\label{fig:np-hard}
	% \vspace{-2mm}
\end{figure}

%\begin{figure}[h]
	%\centering
	%\includegraphics[width=0.6\linewidth]%{ figs/unsub example.pdf}
	%\caption{{Example of non-submodular}}
	%\label{fig:non-submodular}
%\end{figure}

\begin{theorem}
    The trussness gain $TG(\cdot)$ is not submodular.
\end{theorem}

\begin{proof}
    If $TG(\cdot)$ is submodular, for arbitrary anchor set $A$ and $B$, we have $TG(A,G) + TG(B,G) \geq TG(A\cup B,G) + TG(A\cap B,G)$. Now we consider the graph in Fig. \ref{fig:te}(a) with two anchor sets 
    $A = \{(v_3,v_8)\}$ and $B = \{(v_5,v_6)\}$. We have $TG(A,G) + TG(B,G) = 0$ and $TG(A\cup B,G) + TG(A\cap B ,G) = 3$, because when we anchor both A and B, the trussness of dotted edges can increase to 4. Thus, the $TG(\cdot)$ is not submodular.
\end{proof}

\section{Solution}
\label{sec:skc}

%In this section, we first present a baseline algorithm, which iteratively selects the best anchor edge, 
In this section, we begin by introducing a baseline algorithm that iteratively chooses the optimal anchor edge,
\ie the edge that can bring the highest trussness gain.
Then two advanced techniques are proposed in Section~\ref{sec:afc} and Section~\ref{sec:fct} to enhance the baseline algorithm.

\subsection{Baseline}\label{sec:base}

The inherent complexity of the problem makes exact solutions very time-consuming, therefore, we develop a heuristic greedy algorithm.
As shown in Algorithm~\ref{alg:baseline}, we iteratively choose the edge with the highest $TG(\{e\},G_A)$ as the anchor (lines 2-5). 
In each iteration, to compute $TG(\{e\},G_A)$ for each edge $e$ in line 3, we utilize the truss decomposition (\ie Algorithm~\ref{alg:truss-decomp}) on $G_{A \cup \{e\}}$ to calculate the trussness gain (line 3). 
This process is repeated for $b$ iterations to obtain the anchor edge set $A$.

%While the greedy approach presents an expedited solution pathway for our problem, the expansive search space and the frequent truss decomposition to compute trussness gain make the basic greedy algorithm computationally intensive and time-consuming.

%The hardness of the problem motivates us to develop novel algorithms to solve the problem. Firstly we develop the greedy algorithm (algorithm \ref{alg:baseline}) which iteratively choose the best edge that brings the most trussness gain. The time complexity of the greedy algorithm is $O(b*m^{2.5})$. Because in each iteration, we need to compute the trussness gain for every edges by truss decomposition. Different from anchor vertexes, there are far more edges than vertexes. Such a time complexity is very time-consuming even for relatively small graphs. And the large social network makes the greedy algorithm too time-cost. 

%We find that with each edge anchored, the trussness change is restricted and the trussness gain computation can also be largely reduced. anchoring most edges can't bring trussness gain in fact. To improve the efficiency of the greedy algorithm, we aim to reduce the number of candidate anchors and the cost of compute the trussness gain.

\begin{algorithm}[t]
{
    \SetVline
    \footnotesize
    \caption{\textbf{Greedy algorithm}}
    \label{alg:baseline}
    \Input{$G$ : the graph, $b$ : the budget}		
    \Output{$A$ : the set of anchor edges}

    \State {$A \gets \varnothing$}
    \While {$|A| < b$}
    {
        \State{\textbf{for each} $e \in E \backslash A$ \textbf{do} compute $TG(\{e\},G_A)$}
        \State{$e^* \gets \mathop{\argmax}\limits_{e \in E \backslash A} TG(\{e\},G_A)$}
        \State{$A$ $\gets A \cup \{e^*\}$}   
    }
    \Return{$A$};
}
\end{algorithm}

While the greedy approach presents an expedited solution for the problem, the frequent use of truss decomposition to calculate truss gain makes the algorithm still unsuitable for larger-scale networks.
In line 3,
we apply truss decomposition on the entire graph to compute the trussness gain for an edge $e$, whose time complexity is $O(m^{1.5})$. 
However, we observe that only a few edges will increase their trussness after anchoring $e$, making it wasteful to recalculate the trussness for the entire graph. 
Moreover, in each iteration, we need to compute the trussness gain for each edge in the graph to find the best one, whose time complexity is $O(m^{2.5})$.
But after anchoring an edge, the trussness gain for most edges does not change. 
This redundant process is repeated for $b$ rounds to extract the anchor set $A$. The overall time complexity is $O(b \cdot m^{2.5})$ which can only be applied on small graphs. 
Thus the subsequent sections are dedicated to enhancing the efficiency of the greedy algorithm by accelerating the trussness gain computation for each edge and reducing the redundancy computation during each iteration.

%\textcolor{red}{The baseline approach has two main drawbacks: $i)$ After anchoring an edge, the number of edges whose trussness increase is restricted, yet line 4 recalculates the trussness for all edges. $ii)$ After selecting an anchor, many results from the previous round remain unchanged, making full re-computation inefficient. The subsequent sections are dedicated to enhancing the efficiency of the greedy algorithm by reducing both the trussness gain computations and the redundancy in these computations during each iteration. }

\subsection{Accelerating trussness gain computation}
\label{sec:afc}
In this section, we present an efficient method to calculate the trussness gain for a selected anchor edge.
For ease of understanding, we discuss this efficient method with $b = 1$, \ie selecting the best anchor edge which leads to the largest trussness gain.
When $b > 1$, this efficient method can be directly used to obtain a new anchor edge by simply treating the $G_A$ as $G$ after each iteration.
We first present some necessary lemmas and concepts.

% Algorithm \ref{alg:baseline} needs to call truss decomposition algorithm to check $TG(.)$ because we are uncertain about: $(i)$  the extent to which trussness can increase after anchoring an edge, and $(ii)$ The influence scope after anchoring an edge. Thus in this part, we carefully consider the deletion order in truss decomposition, and find that it can be used to accelerate follower computation. we first constrain the trussness change after the anchoring of edge $x$, and propose an route-based method to check the influence scope after anchoring an edge.

%We begin by introducing relevant definitions before delving into our techniques. 

%Then the following lemma indicate that trussness change is restricted when anchoring an edge.

% \begin{lemma}\label{lem:increase by 1}
% After the anchoring of edge $x$ in $G$, each edge $e \in E \backslash A$ can increase its trussness by at most 1, \ie $t^{A}(e) - t^{A \backslash \{x\}} (e) \leq 1$.
% \end{lemma}

% \begin{proof}
%     We prove this lemma by contradiction. Suppose there is a non-anchor edge $e \in E \backslash A$ with trussness increasing from $k_1$ to $k_2$ after anchoring edge $x$, where $k_2 > k_1+1$. Let $S$ be the $k_2$-truss after anchoring edge $x$. Thus we have $sup(e,S) \geq k_2 - 2$ for $ e \in S$. Now if we delete the anchored edge $x$, we have  $sup(e,S) \geq k_2 - 3$ for $ e \in S$ because deleting an edge can only break one triangle. Thus $S \backslash x \subseteq  k_2-1$-truss. And $e \in S$ and $x \neq e$, so $e \in k_2-1$-truss and $k_1 \geq k_2-1$ which contradicts $k_2 > k_1+1$.
% \end{proof}

\begin{lemma}\label{lem:increase by 1}
After the anchoring of edge $x$ in $G$, each edge $e \in E$ can increase its trussness by at most 1, \ie $t^{\{x\}}(e) - t(e) \leq 1$.
\end{lemma}

Note that due to space constraints, the proofs of all lemmas in this paper are included in our online appendix \cite{appendix}.

% \begin{proof}
%     We prove this lemma by contradiction. Suppose there is a non-anchor edge $e \in E$ with trussness increasing from $k_1$ to $k_2$ after anchoring edge $x$, where $k_2 > k_1+1$. Let $S$ be the $k_2$-truss after anchoring edge $x$. 
%     Thus we have $sup(e,S) \geq k_2 - 2$ for $ e \in S$. 
%     Now if we delete the anchored edge $x$, we have $sup(e,S\backslash \{x\}) \geq k_2 - 3$ for $ e \in S$ because deleting an edge can only break one triangle. 
%     Thus $S \backslash x \subseteq  (k_2-1)$-truss. 
%     And $e \in S$ and $x \neq e$, so $e \in (k_2-1)$-truss and $k_1 \geq k_2-1$ which contradicts $k_2 > k_1+1$.
% \end{proof}

% The baseline algorithm employs truss decomposition on the entire graph to obtain precise anchor followers, resulting in significant time overhead. 
After anchoring $x$, we refer to the edges whose trussness increases as the \textit{follower} of $x$, denoted by $F(x,G)$, \ie $F(x,G) = \{e \in G | t^{\{x\}}(e) > t(e)\}$.
According to Lemma \ref{lem:increase by 1}, we assert $TG(\{x\},G)=|F(x,G)|$.
Consequently, when an edge $x$ is anchored, the computation of the trussness gain can be translated into determining the number of followers of $x$. 
% Below, we will explain in detail how to quickly identify the followers of an edge.

% According to Lemma \ref{lem:increase by 1}, we can call those non-anchored edges whose trussness increases after anchoring $e$ \textit{followers} of $e$, denoted by $F(e,G)$, \ie $F(e,G) = \{e' \in G | t^{\{e\}}(e') > t(e')\}$, thus the $F(e,G)$ can directly represents the trussness gain after anchoring $e$. However, upon restricting trussness changes, we observe that followers are predominantly concentrated in a small portion of the graph. To circumvent the need for truss decomposition on the entire graph, we introduce a layer-based approach utilizing a triangle-connected upward route. Then, we commence by introducing the concept as following.

\begin{definition}[$k$-hull]
    Given a graph $G$, the $k$-hull of $G$, denoted by $H_k(G)$, is a set of edges with trussness equal to $k$, \ie $H_k(G) = \{e \in G | t(e) = k\}$.
\end{definition}

\begin{figure}[t]
	\centering
	\includegraphics[width=0.8\linewidth]{ 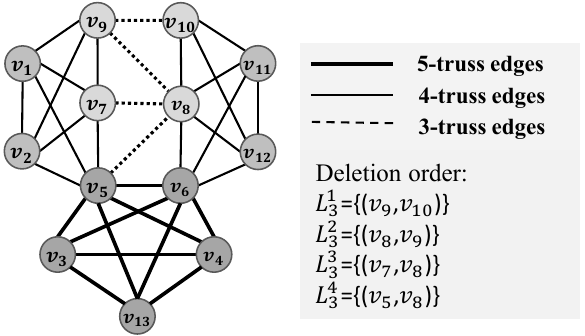}
        %\includegraphics[width=0.49\linewidth]{ figs/example_of_tree.pdf}
	% \vspace{-2mm}
	\caption{{Running example}}
	\label{fig:ame}
	% \vspace{-2mm}
\end{figure}

During truss decomposition (\ie Algorithm~\ref{alg:truss-decomp}), the edges within the $k$-hull are removed in sequential order. 
Specifically, for a given $k$ value, the algorithm removes edges with $sup(e,G) \leq k-2$ in each iteration and modifies the support of edges involved in triangle formations, continuing this process until the support of all remaining edges larger than $k-2$. 
Consequently, based on the deletion order of truss decomposition, we can divide the edges in $k$-hull into several parts (layers).
We use $L_k^i$ to represent the edge set in $k$-hull that is deleted in $i$-th iteration (\ie $i$-th layer), and $l(e)$ to denote the iteration index of $e$, \ie $l(e) = i$ for each $e \in L_k^i$. 
Note that, each edge $e$ has only one iteration index $l(e)$, as it is uniquely associated with only one $k$-hull. 
Given two edges $e_1$ and $e_2$, we define $e_1 \prec e_2$ iff $t(e_1) < t(e_2)$, or $t(e_1) = t(e_2)\ \wedge\ l(e_1) \leq l(e_2)$.

\begin{example}
    In the illustration depicted in Fig. \ref{fig:ame},  the sets of edges with different trussness are $3$-hull, $4$-hull and $5$-hull respectively. We list the $L_k^i$ for edges in the 3-hull. For instance, the support of edge $(v_9,v_{10})$ is $1$, $(v_9,v_{10})$ is deleted in the first round of $4$-truss decomposition. Thus, $L_3^1=\{(v_9,v_{10})\}$. Similarly, we have $L_3^2=\{(v_8,v_9)\}$. Additionally we can easily observe that $(v_9,v_{10}) \prec (v_8,v_9)$.
\end{example}

\begin{definition}[Triangle-connected]\label{def:tri-con}
    Given two edges $e_s$ and $e_t$ in graph $G$, they are triangle connected, if $i)$ $e_s$ and $e_t$ belong to the same triangle, or $ii)$ there exist a series of triangles $\Delta_1, \Delta_2, \cdots, \Delta_j$, such that $e_s \in \Delta_1, e_t \in \Delta_j$, and $\Delta_i \cap \Delta_{i+1} \neq \varnothing$ for $1 \leq i < j$.
\end{definition}

% \begin{figure*}[t]
% 	\centering
% 	\includegraphics[width=0.95\linewidth]{ figs/root edges.pdf}
% 	\caption{Cases of triangles}
% 	\label{fig:tsa}
% \end{figure*}

When $e_s$ and $e_t$ belong to the same triangle, we define $e_s$ and $e_t$ to be \textit{neighbor-edge} of each other.
When $e_s$ and $e_t$ are triangle-connected and do not exist in the same triangle, we can derive an edge set $\{e_s, e_1, e_2, \cdots, e_{j-1}, e_t\}$, where $e_i = \Delta_i \cap \Delta_{i+1}$ for $1 \leq i < j$.
Note that this edge set is order-sensitive.
We refer to this edge set as a \textit{route} from $e_s$ to $e_t$, denoted by $R_{e_s \rightarrow e_t}$.

\begin{definition}[Upward-route]\label{def:route}
    Given two edges $e_s$ and $e_t$ in graph $G$, we say there is an upward-route from $e_s$ to $e_t$, denoted by $R_{e_s \rightsquigarrow e_t}$, if $i)$ there exist a route $R_{e_s \rightarrow e_t} = \{e_s, e_1, e_2, \cdots, e_{j-1}, e_t\}$, $ii)$ $t(e_s) = t(e_i) = t(e_t)$ for $1 \leq i < j$, and $iii)$ $e' \prec e''$ for every two consecutive edges $e'$ and $e''$ along this route.
\end{definition}

% \begin{figure}[h]
% 	\centering
% 	\includegraphics[scale=0.9]{ figs/triangle path.pdf}
% 	\vspace{-2mm}
% 	\caption{{Example of paths}}
% 	\label{fig:triangle path}
% 	\vspace{-2mm}
% \end{figure}

\begin{example} 
    To explain the upward route in Fig. \ref{fig:ame}, we have $R_{(v_9,v_{10}) \rightsquigarrow (v_5,v_8)} = \{(v_9,v_{10}), (v_8,v_9), (v_7,v_8), (v_5,v_8)\}$ because they are triangle-connected. $R_{(v_9,v_{10}) \rightsquigarrow (v_5,v_8)}$ is also a $R_{(v_9,v_{10}) \rightarrow (v_5,v_8)}$ which satisfy condition $i)$ in Definition \ref{def:route}. Then we have $t(v_9,v_{10}) = t(v_8,v_9) = t(v_7,v_8) = t(v_5,v_8) = 3 $ which satisfy condition $ii)$. Finally we have $(v_9,v_{10}) \prec (v_8,v_9) \prec (v_7,v_8) \prec (v_5,v_8)$ which satisfy condition $iii)$.

    %\textcolor{red}{In Figure \ref{fig:k-truss component tree}, We have a route $R_{(b,c) \leftrightarrow (o,p)}:\{ (b,c) \rightarrow (b,h) \rightarrow (g,h) \rightarrow (h,l) \rightarrow (l,m) \rightarrow (l,p) \rightarrow (o,p) \}$ due to the triangle-connectivity. If the edge $(b,c)$ is anchored, we can easily find the root edges $(b,h) ,(c,h)$ which satisfy condition $(c)$ in Figure \ref{fig:tsa}. we also have $R_{(b,h) \rightsquigarrow (l,h)}: \{ (b,h) \rightarrow (g,h) \rightarrow (h,l) \}$ which is a sub-route of $R_{(b,c) \leftrightarrow (o,p)}$. We have $t(b,h)=t(g,h)=t(h,l)$ and $(b,h) \prec (g,h) \prec (h,l)$ which satisfy condition $ii)$ in definition \ref{route}. }

    %A Triangle edge path starts from $(b,c)$ and follows the arrows to reach the end edge $(o,p)$. For any two adjacent edges in this path, we can find two adjacent triangles. For instance, we have $\Delta_{bgh}\cap \Delta_{ghi}=(g,h)$ for adjacent edges $(b,h)$ and $(g,h)$. And the path indicated by the green arrow is a Triangle Connected edge Upstair Path. One is from $(b,c)$ to $(h,l)$ while another is from $(b,c)$ to $(h,q)$. For the path from $(b,c)$ to $(h,l)$, we have (1) $w(\Delta_{bgh})=w(\Delta_{ghl})=3\geq w(\Delta_{bch})$. And the first triangle $\Delta_{bch}$ is a strong triangle. (2) $t(b,c)=t(b,h)=t(g,h)=t(l,h)=3$. (3) $(b,c)\prec (b,h)\prec (g,h)\prec(h,l)$.
\end{example}

\begin{lemma}\label{lem:ep}
    If edge $e_t \in G$ is a follower of the anchor edge $x$ (\ie $e_t \in F(x,G)$), one of the following necessary conditions must be satisfied: $i)$ $e_t$ is a neighbor-edge of $x$ where $t(e_t) > t(x)$ or $t(e_t) = t(x) \wedge l(e_t) > l(x)$; $ii)$ there exist an upward-route $R_{e_s \rightsquigarrow e_t}$ where $e_s \in F(x,G)$ is a neighbor-edge of $x$  and satisfy condition (1).
\end{lemma}

In this paper, we define an edge that satisfies any of the necessary conditions in Lemma~\ref{lem:ep} as a \textit{candidate follower} of an anchor edge $x$.
Thus, to find the true followers of an anchor edge $x$ (\ie $F(x,G)$), instead of performing truss decomposition on the entire graph, we only need to focus on its candidate followers.
For a candidate follower $e$ of an anchor edge $x$, if $e \in F(x,G)$, there must be $t(e) - 1$ triangles containing $e$ in $T_{t(e)+1}(G_{\{x\}})$.
However, before obtaining $F(x,G)$, it is not possible to get the exact number of triangles in $T_{t(e)+1}(G_{\{x\}})$ that contains $e$.
This is because edges in $F(x,G)$ will form new triangles with $e$ in $T_{t(e)+1}(G_{\{x\}})$.
When obtaining $F(x,G)$, each candidate follower $e$ has three status: \textbf{unchecked}, \textbf{survived} and \textbf{eliminated}.
An edge is considered \textbf{unchecked} if it has not been evaluated against the support constraint.
%An edge is set as \textbf{unchecked} if it has not been checked with support constraint.
%An edge is set as \textbf{survived} if it survived the support check, otherwise it becomes \textbf{eliminated}.
An edge is labeled \textbf{survived} if it passes the support check; otherwise, it is \textbf{eliminated}.
Then, we introduced a concept of the effective triangle to capture the potential triangles that can support $e$ stay in $T_{t(e)+1}(G_{\{x\}})$.

\begin{definition}[Effective triangle]\label{def:et}
    Given a triangle consisting of edges $e$, $e_1$ and $e_2$, we say this triangle is an effective triangle of $e$ if $i)$ $e_1$ and $e_2$ are not \textit{eliminated}, $ii)$ $e \prec e_1$ or $e_1$ is \textit{survived}, and $iii)$ $e \prec e_2$ or $e_2$ is \textit{survived}.
\end{definition}

Let $s^+(e)$ be the number of effective triangles of an edge $e$.
We utilize $s^+(e)$ as an upper bound for $sup(e, T_{t(e)+1}(G_{\{x\}}))$.
The following lemma establishes that a candidate follower $e$ can be disregarded if its support upper bound is insufficient.
%The following lemma \textcolor{red}{shows that a candidate follower $e$ can be safely eliminated if its support upper bound is inadequate.}
The removal of an edge $e$ may lead to the deletion of additional edges, as described in Algorithm~\ref{alg:FindFollowers}.
%The removal of an edge $e$ may \textcolor{red}{trigger} the deletion of other edges, as \textcolor{red}{illustrated} in Algorithm~\ref{alg:FindFollowers}.
Once the deletion cascade concludes, the status and support upper bound of all edges influenced by the removal of $e$ are correctly updated.
%When the deletion cascade terminates, the status and support upper bound of all the edges affected by the removal of $e$ \textcolor{red}{can} be correctly updated.

% \textcolor{red}{To better explain the definition, we illustrate all types of effective triangles involving $(u,v)$ with different trussness and layer in Fig. \ref{fig:et}. In condition $(a)-(b)$, $t((u,w)) \prec (v,w)$ (resp. $(v,w)$). In $(c)$, $(u,w)$ is a survived edge and $(u,w) \prec (u,v)$, similarly, in $(d)-(f)$, edge $(u,w) \prec (u,v)$ or is a survived edge (resp. $(u,w)$). These effective triangles provide support and may contribute to an increase in $(u,v)$. We define the number of effective triangles of an edge $e$ as its effective support, denoted by denoted as $s^+(e)$. When $s^+(e) \geq t(e) -1$, the edge may increase its trussness, and we classify it as a survived edge. In contrast, if this condition is not met, the edge is classified as an eliminated edge. Our main idea is to use BFS to perform a layer-by-layer check. We start by setting the anchor edge as survived, then proceed from the root edges to evaluate the edges along the path. The detail of the algorithm in show in algorithm \ref{alg:FindFollowers}. }

% \begin{figure*}[t]
% 	\centering
% 	\includegraphics[width=0.9\linewidth]{figs/effective tirangle.pdf}
% 	\vspace{-2mm}
% 	\caption{Six types of effective triangles}
% 	\label{fig:et}
% 	\vspace{-2mm}
% \end{figure*}
\begin{lemma}\label{lem:cf}
   A candidate follower $e$ cannot be the true follower of $x$ if $s^+(e) < t(e)-1$.
\end{lemma}

\begin{algorithm}[t]
    \SetVline
    \footnotesize
    \caption{\textbf{GetFollowers}$(G, x)$}
    \label{alg:FindFollowers}
    \Input{$G$ : the graph, $x$ : the anchor edge}		
    \Output{$F$ : the follower set of $x$}
    
    \State{$x$ is set survived}
    \State{$F, H_3, H_4, \cdots, H_{k_{m}} \gets \varnothing$}
    \State{$NE \gets$ all neighbor-edges of $x$ which satisfy condition $i)$ in Lemma~\ref{lem:ep}}
    % \ForEach{ $e \in RT $}
    % {
    %     \State{$H_{t(e)}.push(e)$} 
    % }
    \State{\textbf{for each} $e \in NE$ \textbf{do} $H_{t(e)}.push(e)$}
    \ForEach{$i$ from $3$ to $k_{m}$}
    {
        \State{set all edges with trussness smaller than $i$ be \textit{eliminated}}
        \While{$H_i \neq \varnothing$}
        {
            \State{$e \gets H_i.pop()$}
            \State{compute $s^+(e)$}
            \If{$s^+(e) \geq t(e)-1$}      
            {
                \State{$e$ is set \textit{survived}}
                \ForEach{neighbor-edge $e'$ of $e$} 
                {
                    \If{$t(e') = i$ \AND $e \prec e'$ \AND $e' \notin H_i$}
                    {
                        \State{$H_i.push(e')$}
                    }
                }
            }
            \Else
            {
                \State{$e$ is set \textit{eliminated}}
                \State{\textbf{Retract}$(e)$}
            }
        }
        \State{put the survived edge set except the anchor into $F$}
    }
    \Return{$F$};

    \vspace{2mm}
    
    {\textbf{Function} \textbf{Retract}$(e)$\\
        \ForEach{survived edge $e'$ which is neighbor-edge of $e$}
        {
            \If{the triangle containing $e$ and $e'$ is a effective triangle of $e'$}
            {
                \State{$s^+(e') \gets s^+(e')-1$}
                \If{$s^+(e') <  t(e')-1$}
                {
                    \State{$e'$ is set \textit{eliminated}}
                    \State{\textbf{Retract}$(e')$}
                }
            }
        }
    }
\end{algorithm}

Algorithm~\ref{alg:FindFollowers} outlines the procedure for calculating the followers of an anchor edge $x$.
%Algorithm~\ref{alg:FindFollowers} shows the details of computing the followers of an anchor edge $x$.
We begin by setting the anchor edge $x$ as \textit{survived} and initializing follower set $F$ and $k_m-2$ min-heaps to store edges (lines 1-2). 
$k_m$ is the largest trussness of the edge in the graph.
% Each min-heap is sorted by the edges' layers. 
For each edge $e$ that satisfies condition $i)$ in Lemma \ref{lem:ep}, we push it into the corresponding min-heap $H_{t(e)}$ (lines 3-4).
The key of an edge in $H$ is its layer number $l(e)$.
% All edges that satisfy condition $i)$ in Lemma \ref{lem:ep} are stored in a set $NE$. (line 3). 
% Note that $k_m$ is the maximal trussness of these edges. 
% Based on their trussness, we push them into the corresponding min-heap $H_{t(e)}$ (line 4). 
Afterward, we perform a layer-by-layer search on each min-heap $H_i$, checking edges along the route until all heaps are empty (lines 5-18). 
All edges with $t(e)<i$ are set to \textit{eliminated} because these edges cannot increase their trussness to $i+1$ according to Lemma~\ref{lem:increase by 1} (line 6). 
When $H_i$ is non-empty, we pop a edge $e$ with minimum $l(e)$ from $H_i$ and compute $s^+(e)$ (line 8-9). 
If $s^+(e) \geq t(e)-1$, the edge $e$ is marked as \textit{survived} (lines 10-11), and we push the candidate followers in the neighbor-edges of $e$ into heap $H_{i}$ (lines 12-14). 
Otherwise, the edge $e$ is set to \textit{eliminated}, and we call function \textbf{Retract}$(e)$ to recursively delete survived edges that no longer have sufficient effective triangles (lines 16-17), which details are shown in lines 21-26.
% Specifically, for each 
% After all the heaps are empty, the survived edges, excluding the anchor, are the exact followers of $x$.
 
 %After that we begin to search different triangle connected upward route with support check (line 8-20). During the process, we first pop the head edge and compute $s^+(e)$ by counting triangles in figure \ref{fig:effective triangle} (line 9-10). Then if $s^+(e) > t(e) - 1$, indicating a potential increase in $t(e)$, we designate it as survived and continue BFS search (line 11-17). It's important to note that surviving edges may later be eliminated if eliminated edges are identified in future iterations. Therefore if $s^+(e)$ don't meet such condition, we must execute retract process in algorithm \ref{alg:Retract(e)} to recursively adjust $s^+(e)$ (line 18-20). After each $H$ become empty, the survived edges, excluding the anchor edge, are precisely the followers (line 21).

% In this process, we need to count the triangles for an edge in the upward path at most three times  in the worst case: $i)$ using BFS to search the upward path (lines 3, 12-13). $ii)$  computing $s^+(e)$ (line 8). $iii)$ determining the collapse of an survived edge (line 16). 
\myparagraph{Complexity analysis}
We require $O(d_{max})$ to find their neighbor-edge in the worst case where $d_{max}$ is maximal value of $d_u+d_v$ of $(u,v)$. 
For each edge in the route, we need to process such operation at most 3 times: $i)$ finding route by using BFS (line 12-14). $ii)$ support check (line 9) $iii)$ retract process (line 17).
Then if given that there are $|E_r|$ edges in the route, the overall time complexity of algorithm \ref{alg:FindFollowers} is $O(3 \cdot |E_r| \cdot d_{max})$, simplified to $O(|E_r| \cdot d_{max})$.

%The time complexity of algorithm \ref{alg:FindFollowers} is $O(m^{1.5})$ because each triangle along the path can be visited at most three times: Find the inner layer edges and push into $H$, support check, and retract to adjust effective support. 

\begin{example}%explain the algorithm of finding followers
Continuing with the same graph in Fig. \ref{fig:ame}, we discuss the algorithm \ref{alg:FindFollowers} when the anchor edge is $(v_9,v_{10})$.  We first set $(v_9,v_{10})$ as a survived edge and collect edges which satisfy condition $i)$ in Lemma \ref{lem:ep}, then we have $H_3=\{(v_8,v_9)\}$ and $H_4=\{(v_8,v_{10})\}$. We first process $H_3$ by using BFS and set edges as eliminated if their trussness is smaller than $3$. After calculating the effective triangles for edge $(v_8,v_9)$. We have effective triangles $\Delta_{v_8v_9v_{10}}$ and $\Delta_{v_7v_8v_9}$ which implies $s^+(v_8,v_9)=2 \geq t(v_8,v_9)-1$. Thus we set $(v_8,v_9)$ as survived edge and continue searching the next consecutive triangle. And we find that $(v_7,v_8)$ satisfies condition $ii)$ in Lemma \ref{lem:ep}, and then push into $H_3$. Now we have $H_3=\{(v_7,v_8)\}$.  Then $s^+(v_7,v_8)=2$ and we set it as survived. Next round $H_3=\{(v_5,v_8)\}$ and $s^+(v_5,v_8)=2$. 
Finally, all dotted edges become followers and we collect them. 
After that, we empty the survived set except the anchor. 
Then we process the $H_4$ and set edges as eliminated if their trussness is smaller than $4$. The effective triangle of $(v_8,v_{10})$ are $\Delta_{v_8v_{10}v_{12}}$ and $\Delta_{v_8v_{10}v_{11}}$ thus $s^+(v_8,v_{10})=2 <t(v_8,v_{10})-1$ thus we set it as eliminated and stop searching. 
There are no followers from this route.

%Firstly we collect root edges, and add them into corresponding min heap $H$. After collecting, we have $H_{3}= \{ (b,h),(c,h) \} $ and begin to use BFS search. We set anchor edge $(b,c)$ as survived at first. Then we pop the edge $(b,h)$ from $H_{3}$. And we compute $S^+(b,h)=2 \geq t(b,h)-1=2$. So we set $(b,h)$ as survived. Thus we visit $\Delta_{bgh}$ and find that $(g,h)$ $\prec$ $(b,h)$, so we add $(g,h)$ to $H_{3}$. Now we have $H_3=\{(c,h),  (g,h)\}$. And then we do the same thing for the edge $(c,h)$ first, because $l(c,h)) < l((g,h))$. After finish processing second layer edges, we have $H_{3}=\{(g,h),(h,q)\}$ and $\{(b,c),(b,h),(c,h)\}$ as survived edges. And after processing third layer, we have $H_{3}=\{(h,l)\}$ and survived ed$\{(b,c), (b,h), (c,h), (g,h),(h,q)\}$. Finally, all the green edges become followers of $(b,c)$.
\end{example}

\subsection{Reducing redundancy computation} \label{sec:fct}

In the first round of the greedy algorithm, the follower set for each edge will be obtained to identify the best edge as the anchor.
Some follower results can be reused in subsequent iterations to avoid redundant computations.
In this section, we introduce a novel structure to decide whether the follower set of an edge $e$ keeps the same in the next iteration.

% \textcolor{red}{Although algorithm \ref{alg:FindFollowers} can significantly reduces cost of finding followers, selecting an edge to anchor alters the graph, necessitating the re-computation of follower results for each edge in the next iteration. Nevertheless, anchoring an edge can only affect trussness along the upward route, which comprises a minuscule portion of the graph. Re-computing follower results for every edge in each iteration is time-consuming, particularly with large graphs and large budgets. To identify reusable follower results and mitigate extensive re-computation, we introduce a structure called the truss classification tree which can classify edges to different part. we first introduce some related definitions. }

\begin{definition}[$k$-truss component]
Given a graph $G$, a subgraph $S$ is a $k$-truss component if $i)$ $S$ is a $k$-truss, and $ii)$ any two edges in $S$ are triangle-connected.
\end{definition}

\begin{table}[t]

    \centering
    \footnotesize

     \caption{Notations for tree structure}
     \renewcommand\arraystretch{1.2}
    \setlength\tabcolsep{3pt}
    \begin{tabular}{|c|p{7cm}|}
     \hline
        \textbf{Notation} & \textbf{Definition} \\ \hline\hline

        $Tc(T_k(G))$ &\ The $k$-truss component of $T_k(G)$ \\ 
        \hline
        
        $\mathcal{T}$ &\ The $k$-truss component tree structure \\ 
        \hline
        $\mathcal{T}[e]$ &\ The tree node that containing edge $e$\\ 
        \hline
        $TN$ &\ A tree node \\ 
        \hline
        $TN.K$  &\ The trussness value associated with tree node \\ 
        \hline
        %$TC(TN)$ &\ Subtree that root by tree node $TN$ \\ 
       % \hline
        $TN.E$ &\ The set of edges in tree node $TN$ \\ 
        \hline
        $TN.I$ &\ The smallest edge $id$ from the $TN.E$ \\ 
        \hline
        $TN.P$ &\ The parent tree node of $TN$\\
        \hline
        $TN.C$ &\ The child tree node set of $TN$\\ 
        \hline
        %$la[e][id]$ &\ Edge $e$'s triangle neighbor in tree node that $TN.I=id$ \\ 
        % \hline
        $sla(e)$ &\ The tree node $id$ set where $id$ $\in$ $sla(e)$ iff exist a neighbor-edge $e'$ of $e$ with $t(e’) \geq t(e)$ and $\mathcal{T}[e'].I=id$\\ 
        \hline
        %$pla(e)$ &\ A $id$ set that $id$ $\in$ $pla(e)$ iff $e’$ $\in$ $TrN(e,G)$ having $t(e’)$ $\textless$ $t(e)$ and $TN.I=id$.\\ \hline
        $F[e][id]$ &\ The followers of edge $e$ in tree node $TN$ with $TN.I=id$\\ 
        \hline
    \end{tabular}
    % \vspace{2mm}
    \label{tab:tree}
    % \vspace{-2mm}
\end{table}

For an integer $k$, %the edges in one $k$-truss component are not triangle-connected to the edges in any other $k$-truss component.
the edges of a $k$-truss component do not form triangle connections with edges from other $k$-truss components.
Additionally, a $k$-truss component is fully contained within a single $(k-1)$-truss component.
%A $k$-truss component is contained \textcolor{red}{within exactly one} $(k-1)$-truss component.
% Different from $k$-truss, the $k$-truss component is a subgraph of $k$-truss in which every $k$-truss component is triangle-connected while they are not triangle-connected between each other. Similar with $k$-truss, a $k$-truss component a subgraph of $k-1$-truss component, $\ie Tc(T_{k}(G)) \subseteq Tc(T_{k-1}(G))$. 
% With $k$ increase, the edges can reasonably be classified to different component. 
Based on the structure relationships between different $k$-truss components, we introduce a tree structure, \ie \textit{truss component tree} ($\mathcal{T}$), to organize the edges in a graph. 
Specifically, $\mathcal{T}$ contains all edges in the graph, with each edge being assigned to a single tree node.
%and each edge is exclusively contained in one tree node.
Given an edge $e$, $\mathcal{T}[e]$ is the tree node that containing $e$.
We then provide a clear description of the tree structure.
%We then clearly describe the tree structure.
Let $TN$ denote a tree node. All edges in a tree node $TN$ have the same trussness value. We use $TN.K$ to represent the trussness value associated with $TN$, and $TN.E$ to denote the edge set $TN$. 
The subgraph induced by the edges in the subtree rooted at $TN$ is a $(TN.K)$-truss component.
Assume that each edge in the graph has a unique identifier, \ie $id$.
We use $TN.I$ to denote the tree node $id$, which is equal to the smallest edge $id$ in $TN.E$.
The hierarchical relationship between tree nodes in the truss component tree is captured by $TN.P$ (the parent tree node of $TN$) and $TN.C$ (the child tree nodes of $TN$).
The notation summarizing this tree structure is provided in Table \ref{tab:tree}.
Then we introduce \textit{subtree adjacency node} to capture the relationship between edges and $TN$.
Given an edge $e$ and a tree node $TN$ in $\mathcal{T}$, we say $TN$ is a subtree adjacency node of $e$ iff exists a neighbor-edge $e'$ of $e$ with $t(e') \geq t(e)$ and $e' \in TN$.
We use $sla(e)$ to store the tree node $id$ $TN.I$ of all $e$'s subtree adjacency node.

\begin{algorithm}[t]
    \SetVline
    \footnotesize
    \caption{\textbf{BuildTree}$(G,RN)$}
    \label{alg:build TCT}
    \Input{$G$ : the graph, $RN$ : a root node}		
    \Output{$\mathcal{T}$ : the truss component tree}

    \State{$G_1, G_2, \cdots, G_l \gets$ the triangle-connected subgraphs in $G$}
    \ForEach{$i$ from 1 to $l$}
    {
        \State{$k_{min} \leftarrow$ the smallest trussness of an edge in $G_i$}
        \State{$TN \gets$ an empty tree node}
        \State{$TN.K = k_{min}$; $TN.P = RN$; $RN.C = RN.C \cup \{TN\}$}
        %\State{\textcolor{red}{Initlize $TN.K,TN.P$;} $RN.C = RN.C \cup \{TN\}$}
        \ForEach{$e \in E(G_i)$ with $t(e) = k_{min}$}
        {
            \State{$TN.E \leftarrow {TN.E} \cup \{e\}$}
            \State{$\mathcal{T}[e] \gets TN$}
            \State{$G_i \leftarrow G_i \backslash \{e\}$}
        }
        \State{$TN.I \leftarrow$ the smallest edge $id$ in $TN.E$}
        \State{$\mathcal{ST} \gets$ \textbf{BuildTree}$(G_i,TN)$}
        \State{$\mathcal{T} \gets {\mathcal{T}} \cup {\mathcal{ST}}$}
    } 
    \Return{$\mathcal{T}$};
\end{algorithm}

Based on the above structure, Algorithm $\ref{alg:build TCT}$ illustrates the details for constructing the truss component tree. 
The construction proceeds from root to leaf, with the input being a social network graph $G$ and an empty root node $RN$. 
Initially, we need to get all triangle-connected subgraphs of $G$ (line 1).
For each subgraph, we identify the smallest trussness value, denoted as $k_{min}$ and initialize the information relevant with $TN$ (lines 3-5). 
Next, the edge classification process then begins (lines 6-9), we assign the edges with trussness equal to $k_{min}$ to the current tree node $TN$. 
These edges are subsequently removed from the subgraph $G_i$ (line 9). 
After processing, the tree node’s $id$ is computed (line 10). 
Finally, for each remaining subgraph $G_i$, we recursively call BuildTree to construct the subtree $\mathcal{ST}$. 
This process continues until $G_i$ becomes empty (lines 11-12). After the whole graph is empty, we finish building the tree.
%Based on above structure, algorithm \ref{alg:build TCT} shows the details of the truss classification tree construction. We make the construction from root to leaf. Note that the input is a social network graph $G$, at the beginning, we process $BuildTree(G,\varnothing)$ to recursively build the tree. We let the $k_{min}$ be the smallest trussness value (line 1). Then we calculate the triangle-connected subgraph of graph $G$ (line 2). Then we begin to classify edges (line3-11). we first create a tree node and add the edges with trussness equal to $k_{min}$ to the tree nodes $TN$, then delete those edges from $G_i$ (line 4-9). After that, we compute the tree node $id$ (line 10). Finally, for the remaining subgraphs $G_i$ within the subgraph, we recursively call $BuildTree$ for constructing the next tree node until $G_i$ is empty (line 11). 

\myparagraph{Complexity analysis}
In each iteration, we need to compute the triangle-connected subgraphs (line 2), which involves visiting all triangles in the graph. 
Since each iteration removes edges with minimal trussness from the graph, the time complexity for this step is bounded by $O(m^{1.5})$. 
Additionally, for each subgraph, we need to visit each edge once to identify those with minimal trussness (lines 6-9), which has a time complexity of $O(m)$. The recursion depth is limited by the maximum trussness value, $k_{\text{max}}$. Hence, the total time complexity is $O(k_{\text{max}} \cdot m^{1.5})$.
%Therefore, the overall time complexity is $O(k_{\text{max}} \cdot m^{1.5})$.
%The time complexity of algorithm \ref{alg:build TCT} is $O(m^{1.5})$ because computing $k$-truss component is similar with truss decomposition. And classifying edges into different tree node only need $O(m)$ time complexity because each edge can only be classified once (line 6-9). So overall time complexity is $O(m^{1.5})$.

\begin{example}
  Proceeding with the example in Fig. \ref{fig:ame}, at the beginning, graph $G$ is triangle-connected, and $k_{min}=3$. 
  We create a tree node $TN_1$ and add edges that trussness equal to $3$ into $TN_1$ and gather relevant information with the node. 
  %Here we use $n_i$ to represent the id of each tree node for simple). 
  Then we delete dotted edges from $G$ and recursively call this function on the remaining subgraph. 
  We have three triangle-connected subgraphs in the next iteration: subgraph induced by node $\{v_1,v_2,v_5,v_7,v_9\}$, $\{v_6,v_8,v_{10},v_{11},v_{12}\}$, $\{v_3,v_4,v_5,v_6,v_{13}\}$. 
  And the $k_{min} = 4$, we create tree node $TN_2,TN_3$ and do the same thing as the last iteration. 
  Round by round, the algorithm returns the tree until the subgraph is empty, the results of Fig. \ref{fig:eot}.
  After finishing build the tree, we have $sla((v_9,v_{10})]) = \{1,14\}$ because in $\Delta_{v_8v_9v_{10}}$, $t(v_9,v_{10}) = t(v_8,v_9)$ and $t(v_9,v_{10}) < t(v_8,v_{10})$. 
  Similarly we have $sla((v_5,v_8)) = \{1,5,14,23\}$.
\end{example}

\begin{lemma}\label{lem: ffsla}
If an edge $x$ is anchored in the graph $G$, we have $F(x) \subseteq \cup_{id \in sla(x)}$ $\mathcal{T}[id].E$.
\end{lemma}

\begin{figure}[t]
	\centering
	\includegraphics[width=1\linewidth]{ 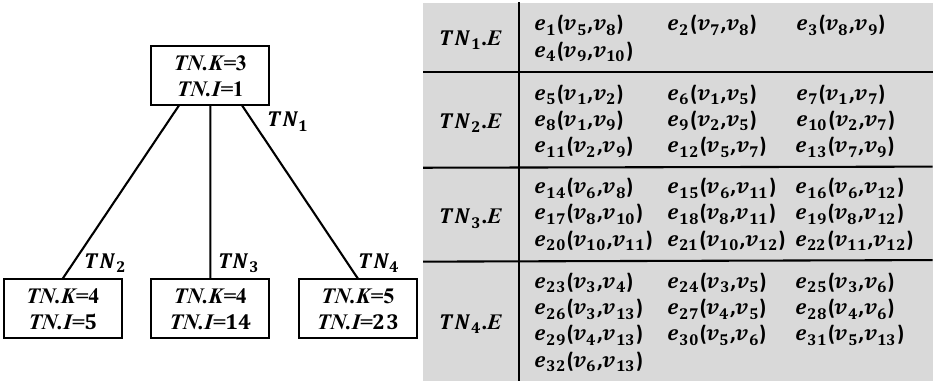}
	% \vspace{-2mm}
	\caption{Example of truss component tree}
	\label{fig:eot}
	% \vspace{-2mm}
\end{figure}

According to Lemma \ref{lem: ffsla}, the followers of an anchored edge $x$ can be divided into several parts based on its subtree adjacency nodes. 
Specifically, we use $F[x][id]$ to record the follower of $x$ in tree node $TN$ with $TN.I = id$, \ie $e \in F[x][id]$ iff $e \in F(x,G)$ and $\mathcal{T}[e].I = id$.
% We organize the followers of edge $x$ from tree node $TN$ using $F[x][id]$, where $TN.I = id$. 
When an edge $x$ is anchored, the trussness of its followers will increase, which results in corresponding modifications to the tree structure. 
If the structure of a tree node $TN$ changes, the follower information of an edge $e$ in $TN$, \ie $F[e][TN.I]$, needs to be updated as well.
For those tree nodes that are not affected, the follower information within them can be reused in the next iteration.

Algorithm \ref{alg:FollowerReuse} shows the pseudo-code for getting the reusable information.
A set $ES$ is used to store the $id$ of the tree node whose structure may change and is initialized to $\{\mathcal{T}[x].I\}$ (line 1).
% Initially, we define the expired set (ES) and set $TN.I$ as expired for $x \in TN$ because we choose the anchor from it. (line 1). 
Then, we collect the $id$ of the tree node where $x$'s followers are located (lines 2-4).
After that, we perform truss decomposition on the subgraph induced by the edges in the subtree rooted at $\{\mathcal{T}[x]\}$, noting that we do not delete anchor edges (lines 5-6).
Note that all anchor edges are preserved during truss decomposition.
  %For each edge $e$, we copy $sla(e)$, the set of tree nodes associated with edge $e$, into the reusable node result set $rn(e)$ (line 2). 
  %We then collect tree node ids in $sla(x)$, which are affected by the anchor edge $x$ (lines 3-5). 
  %For each edge, we remove expired ids from $rn(e)$ based on the affected tree nodes (lines 6-7). 
  %Then we increase trussness of followers by 1 (line 8).
The subtree rooted at $\mathcal{T}[x]$ will be re-construct by Algorithm~\ref{alg:build TCT} (lines 7-9). 
After this restructuring, the followers of the anchor $x$ are merged into different tree nodes with higher $k$, leading to changes in the tree node structure. 
Therefore, we need to collect the $id$ of those newly affected tree nodes (line 11). 
Finally, we remove all expired $id$s from $rn(e)$; $id$s remaining in $rn(e)$ represent the reusable results for edge $e$, meaning that $F[e][id]$ remains unchanged for all $id \in rn(e)$.
Finally, we obtain the reusable tree node $rn(e)$ by removing the $id$s in $ES$ from the $sla(e)$ (lines 12-13). 
The $id$ in $rn(e)$ means that $e$'s follower in tree node $\mathcal{T}[id]$ keeps the same after anchoring $x$.
  
  %Subsequently, we update the trussness of the followers and rebuild the subtree rooted at $TN[x]$while updating the relevant data structures. As a result of the increase in trussness of the followers, these edges will now be contained in different tree nodes than before. Hence, we perform a similar removal process as in lines 3-7 to eliminate unreusable results in the new tree nodes containing the followers (lines 14-15). Finally, we return the reusable tree nodes for each edge (line 16).

\begin{algorithm}[t]
    \footnotesize
    \SetVline
    \caption{\textbf{FollowerReuse}$(G, x, \mathcal{T})$}
    \label{alg:FollowerReuse}
    \Input{$G$ : the graph, $x$ : the anchor, $\mathcal{T}$ : the truss component tree}	
    \Output{$rn(\cdot)$ : for each $e \in G$ where $F[e][id]$ can be reused for each $id \in rn(e)$}

    \State{$ES \gets \{\mathcal{T}[x].I\}$}
    \ForEach{$id \in sla(x)$}
    {
        \If{$F[x][id] \neq \varnothing$}
        {
            \State{$ES \gets ES \cup \{id\}$}
        }
    }
    
    %\State{\textbf{for each} $e \in F(x,G)$ \textbf{do} $t(e) = t(e)+1$}
    %\State{$G' \gets$ the subgraph induced by the edges in the subtree rooted at $\mathcal{T}[x]$}
    \State{$G' \gets$ the subgraph formed by edges within the subtree rooted at $\mathcal{T}[x]$}
    \State{\textbf{TrussDecomp($G'$)}}
    \State{$P' \gets \mathcal{T}[x].P$}
    \State{$\mathcal{T'} \gets$ \textbf{BuildTree}$(G',P')$}
    \State{$\mathcal{T^*} \gets \mathcal{T}$ with the subtree root at $P'$ replaced by $\mathcal{T'}$}  
    \State{get new $sla(e)$ from $\mathcal{T^*}$ for each $e \in G$}
    \State{$ES \gets {ES} \cup \{\mathcal{T^*}[e].I | e \in F(x,G)\}$} 
    \ForEach{$e\in G$}
    {
        \State{$rn(e) \gets sla(e) \backslash ES$}  
    } 
    \Return{$rn(\cdot)$};
\end{algorithm}

\myparagraph{Complexity analysis}
Algorithm \ref{alg:FollowerReuse} has a time complexity of $O(k_{max} \cdot m^{1.5})$
%The time complexity of Algorithm \ref{alg:FollowerReuse} is $O(k_{max}\cdot m^{1.5})$
, as it relies on the tree-building process, which requires $O(k_{max}\cdot m^{1.5})$ time complexity. 
In the worst case, the entire tree needs to be re-constructed. 
Each edge is visited only once during the expiration check (lines 12-13), contributing an additional $O(m)$ complexity. 
Furthermore, the size of $sla(x)$ is bounded by $m$ (lines 2-4). 
Consequently, the overall time complexity is $O(k_{max}\cdot m^{1.5})$, equivalent to that of the tree construction process.

\begin{lemma}\label{lem:rn}
After anchoring $x$, for each non-anchored edge $e$, $F[e][id]$ remains unchanged if $id \in rn(e)$.
\end{lemma}

\subsection{GAS algorithm} \label{sec:gasa}

Our final greedy algorithm assembles all the above techniques. First, we construct the truss component tree. Then, we identify the edge with most followers in each iteration and utilize the reused information to avoid redundant computation in the next iterations.

\begin{algorithm}[t]
    \footnotesize
    \SetVline
    \caption{\textbf{GAS($G,b$)}}
    \label{alg:GreedyAnchorSelection}
    \Input{$G$ : the graph, $b$ : the budget}	
    \Output{$A$ : the set of anchor edges}

    \State {$A \gets \varnothing$}
    \State{get $t(e)$ and $l(e)$ for each edge $e \in G$ by Algorithm \ref{alg:truss-decomp}}
    \State{$\mathcal{T} \gets $ \textbf{BuildTree}$(G, \varnothing)$}
    \State{\textbf{for each} $e \in G$ \textbf{do} $rn(e) \gets \varnothing$}
    \While{$|A| < b$}
    {   
        \State{$e^* = null$; $Max = 0$}
        %\State{get $|E(e)|$ for $e \in G$}
        \ForEach{$e \in G \backslash A$}
        {
        %\If{$e \notin A$}
        %{
            \ForEach{$id \in sla(e) \backslash rn(e)$}
            {
                \State{$F[e][id] \gets \textbf{GetFollowers}(G, e)$}
                \If{$|F(e,G)| > Max$}
                {
                    \State{$e^* = e$; $Max = |F(e,G)|$}
                }
            }
        %}
        }
        \State{$A \gets  A \cup \{e^*\}$;  $sup(e^*, G) = + \infty$}
        \State{$rn(\cdot) \gets$ \textbf{FollowerReuse}$(G,x,\mathcal{T})$}
    }
    \Return{$A$};
\end{algorithm}

Algorithm \ref{alg:GreedyAnchorSelection} shows the details of our \texttt{GAS} algorithm. 
First, we initialize an empty anchor set $A$, and compute the trussness $t(e)$ and layer $l(e)$ for each edge using the truss decomposition algorithm (lines 1-2). 
The tree structure is constructed by Algorithm \ref{alg:build TCT} (line 3). 
% After that, we initialize $rn(e)$ empty sets, indicating that no results are reusable at the beginning (line 4).
Then, we perform $b$ rounds to obtain the anchor set $A$ (lines 5-13). 
$e^*$ and $Max$ are used to record the current selected best edge in each iteration and the corresponding number of followers (line 6). 
We compute the followers for each non-anchored edge (lines 7-11).
Note that, we only need to recompute the followers of each edge on the non-reusable tree node $sla(e)\backslash rn(e)$ by a variant of Algorithm \ref{alg:FindFollowers} (line 9), which is equipped with reuse technique.
The only difference is that we ignore routes in reusable tree nodes. 
Specifically, in line 4 of Algorithm \ref{alg:FindFollowers}, we only push $e \in ES$ into $H_{t(e)}$ iff $e \in TN.E$ and $TN.I \notin rn(x)$, \ie the result of $F[x][id]$ is reusable.
After the follower computation of the current iteration, we select the edge $e^*$ with the maximum number of followers as the anchor and set its support $sup(e^*,G)$ to be positive infinity (line 12). 
We invoke Algorithm \ref{alg:FollowerReuse} to restructure the tree structure and update the reusable results for the next iteration (line 13). 
After $b$ iterations, the algorithm returns the anchor set $A$ (line 14).

\begin{table*}[t]
\centering
\footnotesize
\caption{Statistics of datasets and algorithm evaluation with default value} \label{table:datasets}
\renewcommand\arraystretch{1}
\setlength{\tabcolsep}{1.5mm}
\setlength{\extrarowheight}{2pt}
\resizebox{0.9\textwidth}{!}{
    \begin{tabular}{|l|c|c|c|c||c|c|c|c||c|c|c|}
        \hline
        \multirow{2}{*}{\textbf{Dataset}} & \multirow{2}{*}{\textbf{Vertices}} & \multirow{2}{*}{\textbf{Edges}} & \multirow{2}{*}{\textbf{$k_{max}$}} & \multirow{2}{*}{\textbf{$sup_{max}$}} & \multicolumn{4}{c||}{\textbf{Trussness gain}}  & \multicolumn{3}{c|}{\textbf{Running time (seconds)}}  \\ \cline{6-12} 
        &  &  &  &  & \texttt{Rand} & \texttt{Sup} & \texttt{Tur}& \texttt{GAS} & \texttt{BASE} & \texttt{BASE+} & \texttt{GAS} \\ \hline
        
        \textbf{\underline{Col}lege}&1,899 & 13,838& 7 & 74  & 111 & 134 & 184 & \textbf{769} & 98547.74 & 88.91 &  \textbf{76.60} \\   \hline
        \textbf{\underline{Fac}ebook} & 4,039 & 88,234& 97 & 293   & 8,891 & 525 & 9,948 & \textbf{21,980} & - & 17788.76 &  \textbf{3122.52} \\   \hline
        \textbf{\underline{Bri}ghtkite} & 58,228 & 214,078 & 43 & 272  & 1271 &  235 & 1,526 & \textbf{6,163} & - & 3388.98 & \textbf{1054.22} \\   \hline
        \textbf{\underline{Gow}alla} & 196,591 & 950,327 & 29 & 1297  & 577 & 769 & 1,042 & \textbf{11,492}  & - & 24414.38 &  \textbf{6732.54} \\ \hline
        \textbf{\underline{You}tube} & 1,134,890 & 2,987,624 & 19 & 4034  & 358 & 823 & 1,611 & \textbf{10,281} & - & 62391.04 & \textbf{22550.14} \\ \hline
        \textbf{\underline{Goo}gle}& 875,713 & 4,322,051 & 44 & 3086  & 91 & 95 & 147 &\textbf{5,640} & - & 76856.74 &  \textbf{15714.23} \\ \hline
        \textbf{\underline{Pat}ents} & 3,774,768 &16,518,947 & 36 & 591  & 59 & 37 & 146 &\textbf{10,870} & - & 194103.18 &  \textbf{70802.71} \\ \hline
        \textbf{\underline{Pok}ec} & 1,632,803 &22,301,964 & 29 & 5566  & 302 & 436 &  809  & \textbf{28,208} & - & - & \textbf{210571.13}  \\ \hline
    \end{tabular}
}
\end{table*}

\myparagraph{Complexity analysis}
In line 3 of Algorithm \ref{alg:GreedyAnchorSelection}, we need $O(k_{max} \cdot m^{1.5})$ time complexity. In lines 7-11, we need to compute the followers of each edge. 
Given the average degree $d_{ave}$ and average route size $E_{ave}$, we need $O(m \cdot E_{ave} \cdot d_{ave})$ to get all results. 
And we need $b$ rounds to get the anchor set. 
Each round, we also need to decide the reusable results, which need $O(k_{max} \cdot m^{1.5})$ time complexity. 
%Thus the overall time complexity is 
In total, the computational complexity is $O(b \cdot (k_{max} \cdot m^{1.5}+m \cdot E_{ave} \cdot d_{ave}))$.

\section{Experiments}
\label{sec:exp}

% 寻找更合适的数据集 
%scalability 实验
%reuseable result test
%route test实验

% In this section, we evaluate the effectiveness and efficiency of our proposed algorithms on eight real-world social networks.

\subsection{Experiment setup}
\label{sec:setup}

\myparagraph{Algorithms}
%\textcolor{red}{As far as we know, no existing work has addressed our problem.}
To the best of our knowledge, there is no existing work for our problem.
Towards the effectiveness, we compare our greedy algorithm (\texttt{GAS}) with four algorithms (\texttt{Exact}, \texttt{Rand}, \texttt{Sup}, and \texttt{Tur}). We also implement and evaluate three algorithms (\texttt{BASE}, \texttt{BASE+}, and \texttt{GAS}) to verify the performance of each proposed technique. A concise overview of all algorithms is as follows:
%A brief description of all algorithms is as follows:
\begin{itemize}[leftmargin=1em]
    \item \texttt{Exact}: identify the optimal anchor set by exhaustively checking all possible combinations of $b$ edges. 
    
    \item \texttt{Rand}: randomly chooses $b$ anchors from $G$.
    
    \item \texttt{Sup}: randomly chooses the $b$ anchors from $G$ with top 20\% highest support.
    
    \item \texttt{Tur}: randomly choose the $b$ anchors from $G$ with top 20\% upward route size.
    
    \item \texttt{BASE}: the baseline method proposed in Section \ref{sec:base}.
    
    \item \texttt{BASE+}: \texttt{BASE} algorithm equipped with upward route to get followers.
    
    \item \texttt{GAS}: Algorithm \ref{alg:GreedyAnchorSelection} developed in Section~\ref{sec:gasa}.

    \item \texttt{AKT}: The method of anchoring vertices to enlarge corresponding $k$-truss in~\cite{zhangfanefficiently2018}.
\end{itemize}

\myparagraph{Datasets} We use eight real-world datasets in our experiments, with details provided in Table \ref{table:datasets}, which are
%We employ 8 real-world datasets in our experiments, whose details are shown in Table \ref{table:datasets},
listed in increasing order of their edge numbers, where $k_{max}$ and $sup_{max}$ are the maximal trussness value and support respectively. 
All datasets can be found on SNAP\footnote{\url{http://snap.stanford.edu}}.

% \begin{table}[h]
% 	\centering
% 	\footnotesize
% 	\caption{\textbf{Statistics of Datasets}} \label{table:datasets}
% 	\vspace{-2mm}
% 	\renewcommand\arraystretch{1.1}
% 	\setlength{\tabcolsep}{0.9mm}{
% 	\begin{tabu}{|l|[1pt]c|c|c|c|}
% 		\hline
% 		\textbf{Dataset}& \textbf{Vertices}& \textbf{Edges}& \textbf{$k_{max}$} & \textbf{$support_{max}$}         \\
% 		\tabucline[1pt]{-}
% 		\hline
%             ($D1$) College&1,899& 13,838& 7 & 74\\
% 		\hline
% 		% \textbf{Epinion} & 131,828& 711,210& 592,592 & 118,618& 4,410,518& 499,558 \\
% 		% \hline
% 		  ($D2$) Facebook & 4039& 88234& 97 & 293      \\
% 		\hline
% 		  ($D3$) Brightkite & 58,228 & 214,078 & 43 & 272 \\
% 		\hline
%             ($D4$) Gowalla & 196,591 & 950,327 & 29 & 1297   \\
% 		\hline
%             ($D5$) Youtube & 1,134,890 & 2,987,624 & 19 & 4034   \\
% 		\hline
%             ($D6$) Google& 875713 & 4,322,051 & 44 & 3086   \\
% 		\hline
%             ($D7$) Patents & 3774768 &16,518,947 & 36 & 591   \\
% 		\hline
%             ($D8$) Pokec & 1632803 &22,301,964 & 29 & 5566   \\
% 		\hline

% 		% \textbf{Pokec} & 1,632,804& 22,301,965& 15,611,376 & 6,690,589& 17,314,923& 15,242,535\\
% 		% \hline
% 	\end{tabu}}
% \end{table}

\myparagraph{Parameters and workload}
We conduct experiments by varying the anchor set size $b$ from 20 to 100, with 100 as the default value. 
All programs are developed using standard C++.
%All the programs are implemented in standard C++. 
The experiments are conducted on a machine with an Intel(R) Xeon(R) 5218R 2.10GHz CPU and 256 GB memory.

\subsection{Experimental results}

\myparagraph{Exp-1: Evaluation of various algorithms on all datasets} Table \ref{table:datasets} summarizes the effectiveness and efficiency of different algorithms on all datasets with default budget $b=100$. 
``-'' means the algorithm cannot finish in three days. 
Our algorithm achieves the highest trussness gain compared to others, within an acceptable runtime. 
For the three random algorithms, \texttt{Rand}, \texttt{Sup} and \texttt{Tur}, the anchor set is chosen randomly 2000 times, and the maximum trussness gain achieved is reported. 
However, these methods yield small trussness gain because anchoring most edges results in no improvement, and the large search space makes it difficult to identify promising anchor sets.
% \textcolor{blue}{Thus we report the maximal trussness gain within 2000 times and the average trussness gain nearby. Although sometimes it's lucky to find the anchor edge choosing by our greedy, like Random in facebook.}
% In random and support algorithms, we randomly choose anchor set two thousands times and report the maximal trussness gain that they achieve. However, they bring very little trussness gain. That's because the anchoring most edges brings no trussness gain and large search scope makes them hard to choose a promising anchor set. 
In the efficiency experiments, \texttt{BASE} is only able to return results on the College dataset due to its high time complexity of $O(b \cdot m^{2.5})$. 
With the incorporation of the upward-route and support check processes, \texttt{BASE+} computes the trussness gain for each anchor more efficiently, as it avoids performing truss decomposition on the entire graph.
However, \texttt{BASE+} recomputes the results for each anchor in every round, leading to significant time overhead, especially on large datasets.
By leveraging the result reuse technique (Section~\ref{sec:afc}), \texttt{GAS} avoids unnecessary recomputation, resulting in faster execution. Notably, the runtime of \texttt{GAS} is
approximately 20\% of that of \texttt{BASE+} on facebook and google.

\myparagraph{Exp-2: Comparison with \texttt{Exact} algorithm} We perform a comparative study of the \texttt{Exact} algorithm and the \texttt{GAS} algorithm.
%We conducted a comparative analysis between the \texttt{Exact} algorithm and the \texttt{GAS} Algorithm. 
The \texttt{Exact} algorithm exhaustively enumerates all possibilities and select the optimal solution within the budget $b$. However, due to the prohibitively high computational time required by the \texttt{Exact} algorithm, We iteratively extract a vertex and its neighbors to form smaller datasets, stopping when the number of extracted edges approaches 150-250,
%we extract small datasets by iteratively extracting a vertex and all its neighbors, until the number of extracted edges reaches nearly 150-250,
following the method proposed in \cite{QingyuanGobal2020}. For budgets ranging from 1 to 3, we record the average running time and the average trussness gain, with the results presented in Fig. \ref{fig: cmp-exact}. The result indicates that the effectiveness of \texttt{GAS} is at least $90\%$ of \texttt{Exact} when the budget does not exceed 3. 
Notably, the trussness gain percentage of \texttt{GAS} compared to \texttt{Exact} may increase as the budget $b$ increases.
%It is worth noting that the trussness gain percentage of \texttt{GAS} relative to \texttt{Exact} may increase as the budget $b$ grows.

\begin{figure}[t]
    \centering
    \subfigure[Facebook]{\includegraphics[width=0.4\linewidth]{ 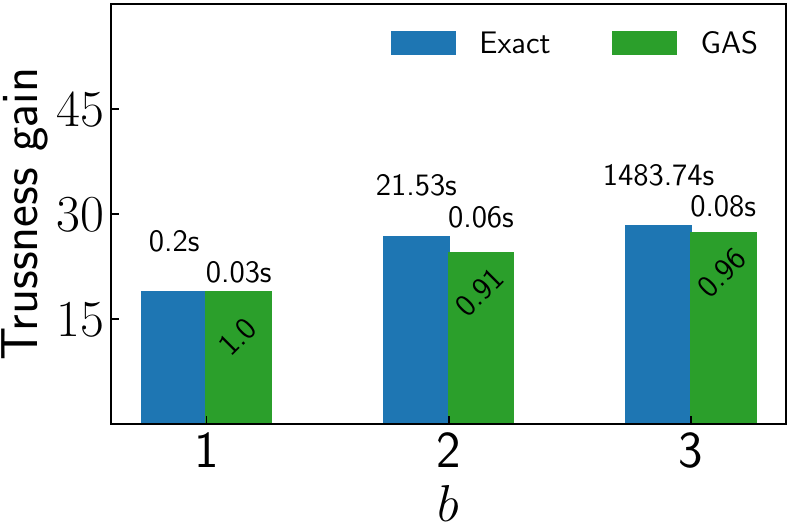}}
    \hspace{5mm}
    \subfigure[Brightkite]{\includegraphics[width=0.4\linewidth]{ 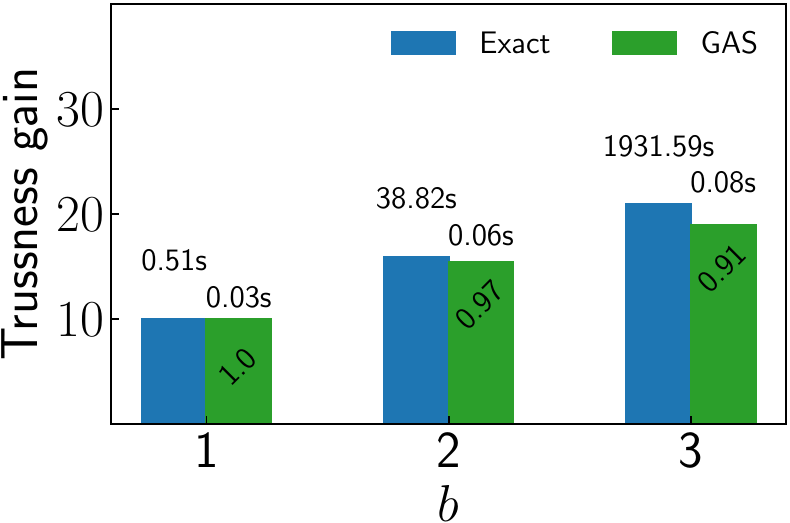}}
    \caption{\texttt{GAS} v.s. \texttt{Exact}}
    \label{fig: cmp-exact}
\end{figure}

\begin{figure}[t]
    \centering
    \subfigure[Facebook]{\includegraphics[width=0.4\linewidth]{ 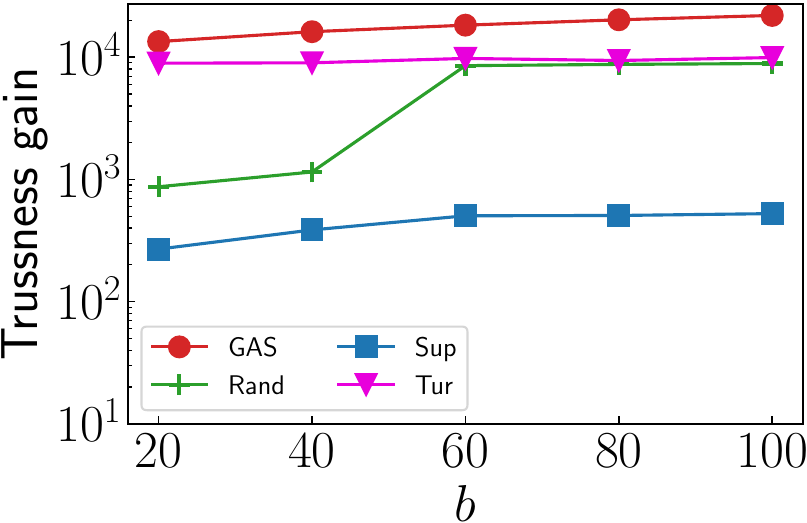}}
    \hspace{5mm}
    \subfigure[Brightkite]{\includegraphics[width=0.4\linewidth]{ 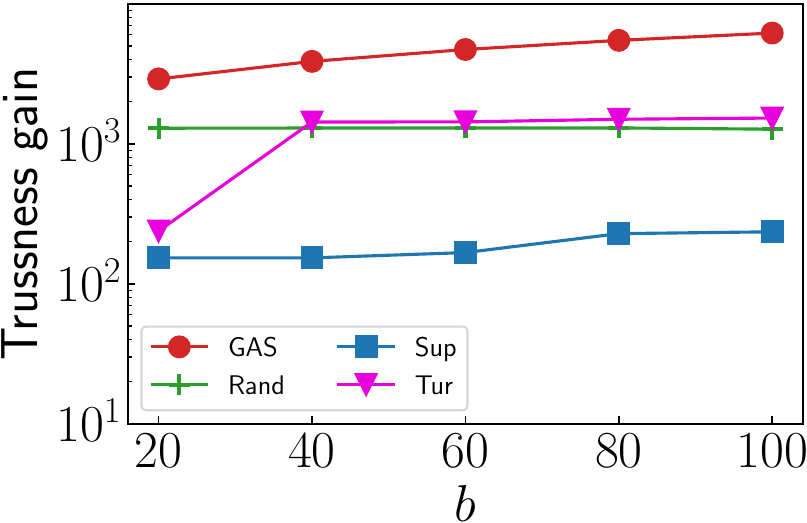}}
    \caption{Effectiveness evaluation by varying $b$}
    \label{fig:TGVB}
\end{figure}

\begin{figure}[t]
    \centering
     \subfigure[\texttt{GAS}]{\includegraphics[width=0.3\linewidth]{ 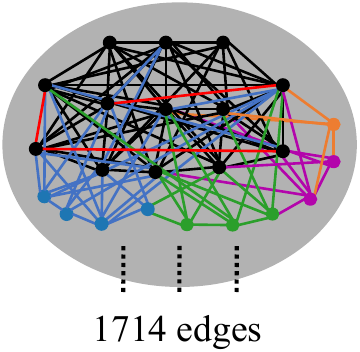}}
     \subfigure[\texttt{AKT}]{\includegraphics[width=0.3\linewidth]{ 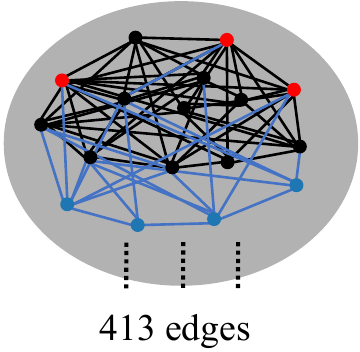}}
     \subfigure[\texttt{Edge-deletion}]{\includegraphics[width=0.3\linewidth]{ 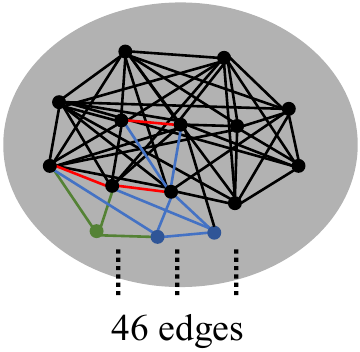}}
    \caption{Case study on Gowalla}
    \label{fig:cs}
\end{figure}

\begin{figure}[t]
    \centering
    \subfigure[College]{\includegraphics[width=0.4\linewidth]{ 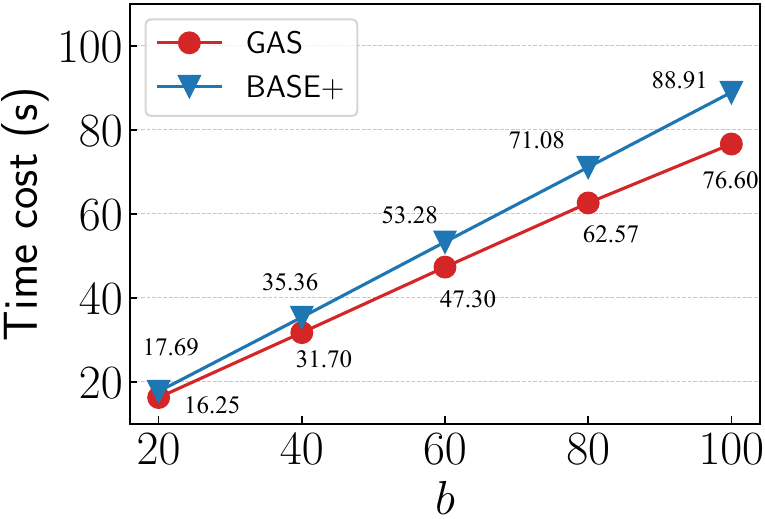}}
    \hspace{5mm}
    \subfigure[Facebook]{\includegraphics[width=0.4\linewidth]{ 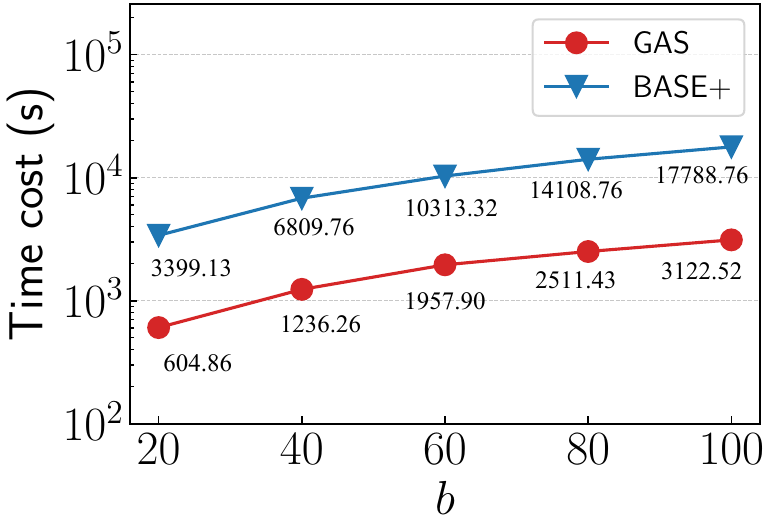}}
    \subfigure[Brightkite]{\includegraphics[width=0.4\linewidth]{ 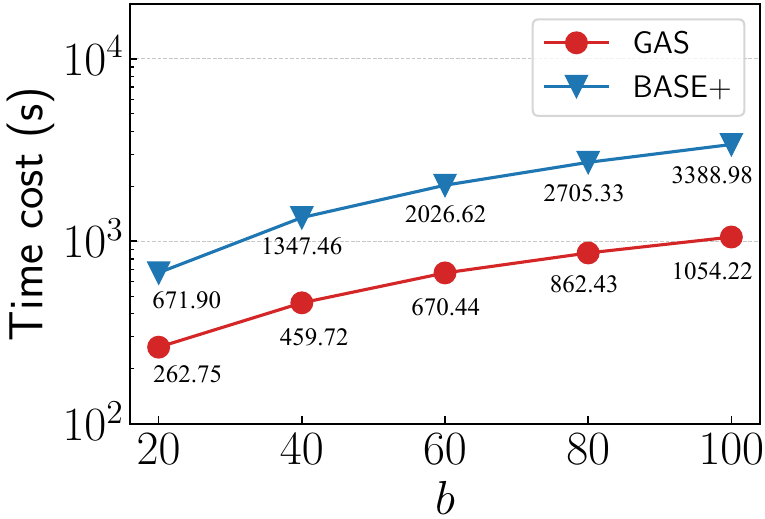}}
    \hspace{5mm}
    \subfigure[Gowalla]{\includegraphics[width=0.4\linewidth]{ 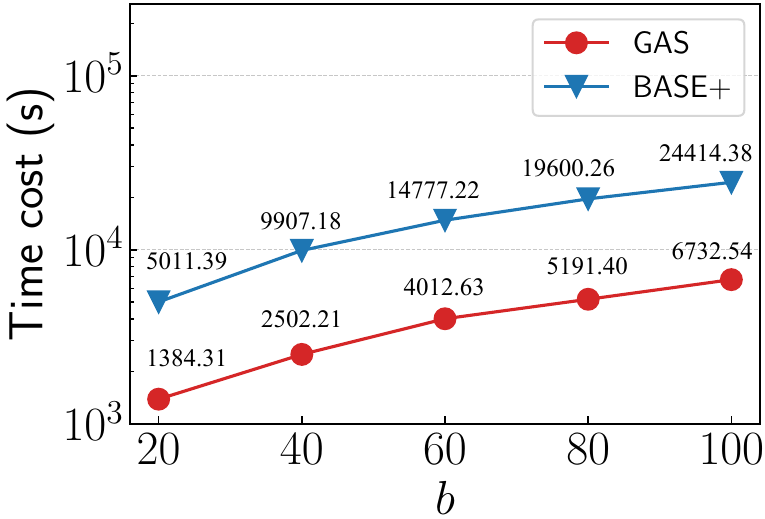}}
    \subfigure[Youtube]{\includegraphics[width=0.4\linewidth]{ 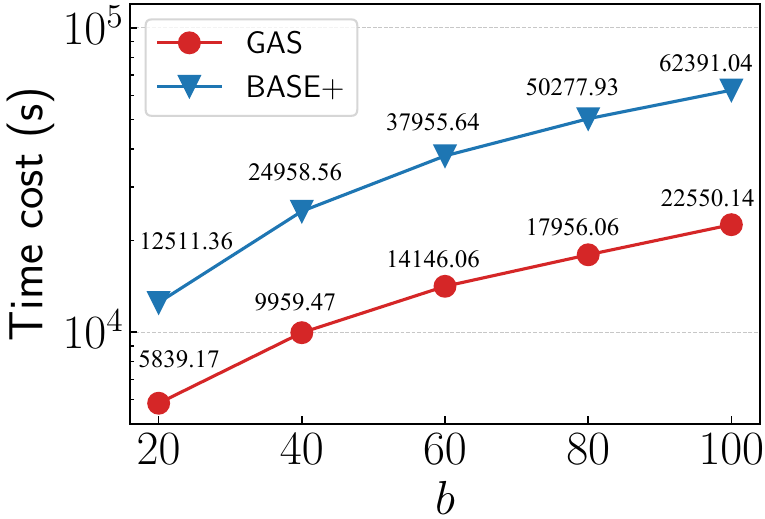}}
    \hspace{5mm}
    \subfigure[Google]{\includegraphics[width=0.4\linewidth]{ 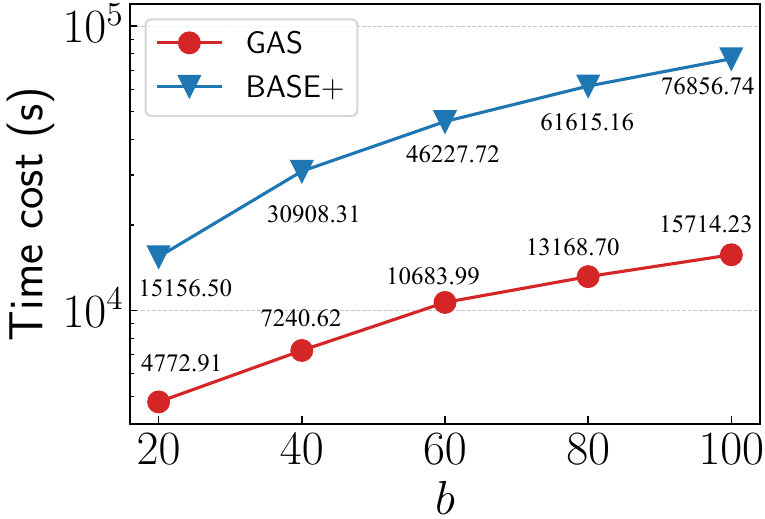}}
    \subfigure[Patents]{\includegraphics[width=0.4\linewidth]{ 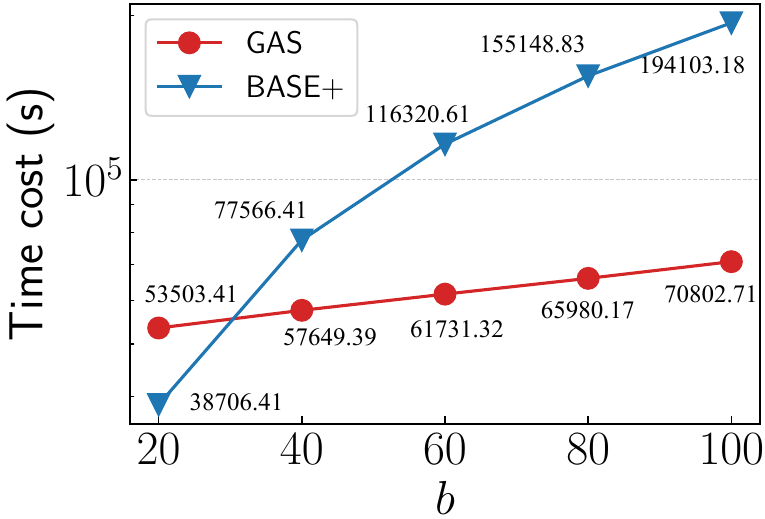}}
    \hspace{5mm}
    \subfigure[Pokec]{\includegraphics[width=0.4\linewidth]{ 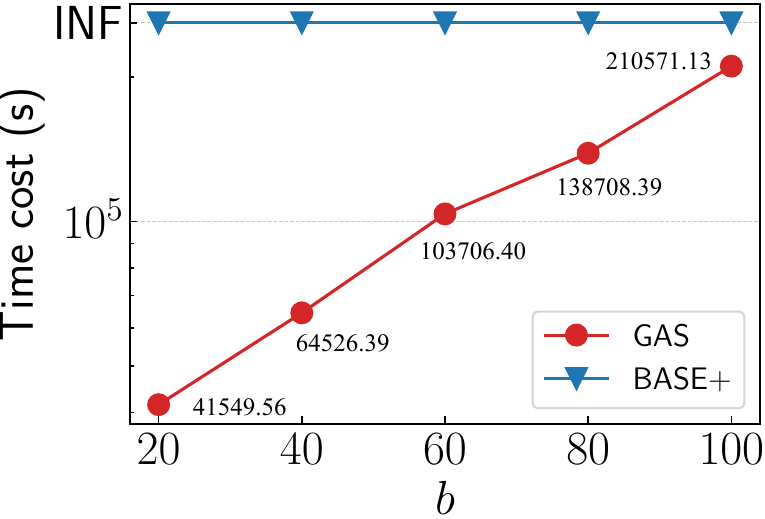}}
    \caption{Efficiency evaluation by varying $b$}
    \label{fig:tvb}
\end{figure}

\myparagraph{Exp-3: Effectiveness evaluation by varying $b$}
%\textcolor{blue}{
In this experiment, we evaluate the trussness gain of the proposed \texttt{GAS} algorithm compared to three random algorithms, \texttt{Rand}, \texttt{Tur}, and \texttt{Sup}, by varying the budget $b$ on Facebook and Brightkite datasets. The results can be found in Fig. \ref{fig:TGVB}. For the three random algorithms, we conducted 2000 independent runs for each budget and reported the maximum achieved trussness gain. The results demonstrate that the \texttt{GAS} consistently outperforms all three random algorithms across all parameter settings, highlighting the effectiveness of our greedy selection strategy. Among the three random algorithms, \texttt{Tur} achieves the best performance, as it selects anchor edges based on the upward-route size, which tends to prioritize edges with a higher potential to acquire more followers. Additionally, \texttt{Rand} outperforms \texttt{Sup}, primarily due to its selection strategy. Specifically, \texttt{Rand} randomly selects anchor edges from the entire graph, whereas \texttt{Sup} restricts the selection to the top 20\% of edges ranked by support. Since edges with high support typically have high trussness, anchoring such edges only benefits other high-trussness edges, while having no impact on lower-trussness edges. In contrast, randomly selecting edges throughout the graph can select anchored edges that have a broader impact on different trussness levels, leading to greater overall trussness gains.
%}

%\myparagraph{Exp-3: Effectiveness evaluation by varying $b$} 
% In this experiment, we examined how the trussness gain varies with the budget on the Facebook and Brightkite datasets. Fig. \ref{fig:TGVB} shows the results of different algorithms under varying budget settings \textcolor{blue}{in which $A_R,A_S,A_T$ stands for the average trussness gain achieved by Rand, Sup and Tur ($M_.$ is the corresponding maximal value).} For the three random algorithms, we performed 2000 independent runs for each budget.
% In the \texttt{Tur} algorithm, anchor edges with longer upward routes have a higher likelihood of generating greater trussness gain. Consequently, selecting anchors from the top 20\% of edges ranked by upward route size produces better results. 
% %Conversely, edges with larger support are less effective as anchors. 
% \textcolor{blue}{It is worth noting that the maximal trussness gain of Random is larger than sup, this is because the optimal anchor edge is not always among the top 20\% edges with the largest support. Giving up the remaining 80\% may lead to the ignorance of the best results, which also shows the advantages of our globally choosing anchors.}
% Overall, the results indicate that our proposed algorithm consistently outperforms the alternatives.

\myparagraph{Exp-4: Case study}
%\textcolor{blue}{
To assess the performance of the proposed \texttt{GAS} algorithm, we conduct a case study comparing it with \texttt{AKT} \cite{zhangfanefficiently2018} and edge-deletion methods. Specifically, \texttt{AKT} selects anchor vertices based on the anchor $k$-truss approach, while the edge-deletion method selects anchor edges as those whose removal leads to the maximum reduction in global trussness. Fig. \ref{fig:cs} shows a case study on the Gowalla dataset with $b=3$, illustrating the trussness gain achieved by these three methods. Since \texttt{AKT} operates on a specific $k$-truss, Fig. \ref{fig:cs} reports the results for the $k$-truss that yields the highest trussness gain. 
In Fig. \ref{fig:cs}, red edges or vertices represent anchors, black edges remain unchanged in trussness, and edges with different colors (except black) indicate trussness increments at different levels. The numbers below the figure denote the count of edges whose trussness has increased. Given the large number of edges in the original graph, directly visualizing the full network structure would obscure the impact of anchoring; thus, Fig. \ref{fig:cs} focuses on highlighting the differences among the three methods. 
From the figure, it is evident that \texttt{GAS} achieves the highest trussness gain compared to the other two methods, enhancing edges across a wider range of trussness levels. In contrast, \texttt{AKT} focuses on a specific $k$-truss subgraph, leading to a lower trussness gain and affecting only edges with trussness equal to $k-1$. On the other hand, edge-deletion selects anchor edges based on their removal impact on global trussness, naturally prioritizing edges with higher trussness. However, since an anchor edge only increases the trussness of edges with an even higher trussness value, the edge-deletion method is less effective in improving global trussness. In summary, these results demonstrate the superiority of the \texttt{GAS} algorithm, as it achieves a significantly higher trussness gain and improves the trussness of edges at various levels across the entire graph.

\myparagraph{Exp-5: Efficiency evaluation by varying $b$} In this experiment, we analyze how the algorithms' runtime varies with the budget size in Fig. \ref{fig:tvb}.
Specifically, we compare the performance of \texttt{BASE+} and \texttt{GAS}. 
Since the \texttt{BASE} fail to return answers within three days on most settings, we excluded it from the results. 
For \texttt{GAS}, constructing the classification tree to facilitate result reuse incurs an initial overhead. Consequently, \texttt{GAS} runs slower at the beginning on the Patents dataset.
However, this initial investment proves to be worthwhile. As shown in Fig. \ref{fig:tvb}, our \texttt{GAS} algorithm consistently delivers results more efficiently than \texttt{BASE+} across all datasets.

%We also do experiments on how the running time of the algorithm varies with budget. we also compare the \textbf{baseline}, \textbf{GAS-R} and \textbf{GAS} algorithms. The baseline algorithm return answers in 43814s even on college dataset with budget 50. On other datasets, it can't return answer within a day, thus we don't report it on figure \ref{fig:TR(.)_budget}. it can be seen that our GAS algorithm return results efficiently.

\begin{figure}[t]
    \centering
    \subfigure[Vary $|E|$: time]{\includegraphics[width=0.43\linewidth]{ 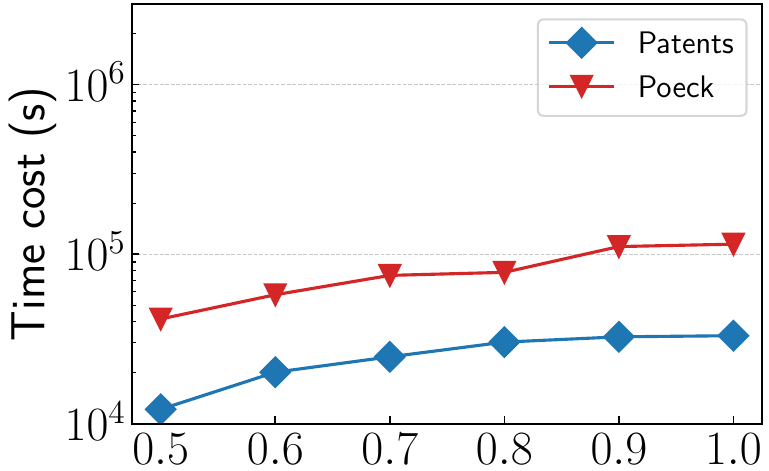}}
    \hspace{5mm}
    \subfigure[vary $|E|$: vertex ratio]{\includegraphics[width=0.43\linewidth]{ 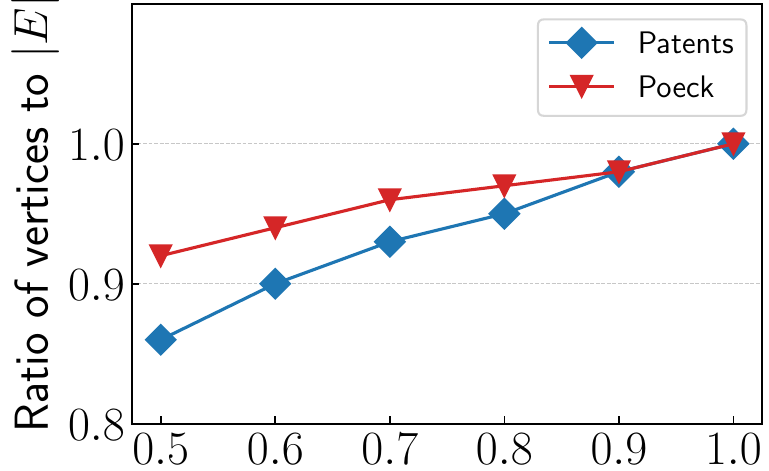}}
    \subfigure[Vary $|V|$: time]{\includegraphics[width=0.43\linewidth]{ 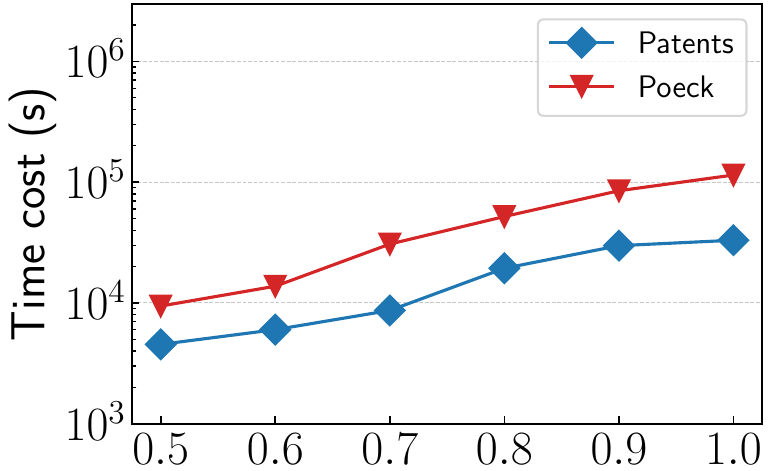}}
    \hspace{5mm}
    \subfigure[Vary $|V|$: edge ratio]{\includegraphics[width=0.43\linewidth]{ 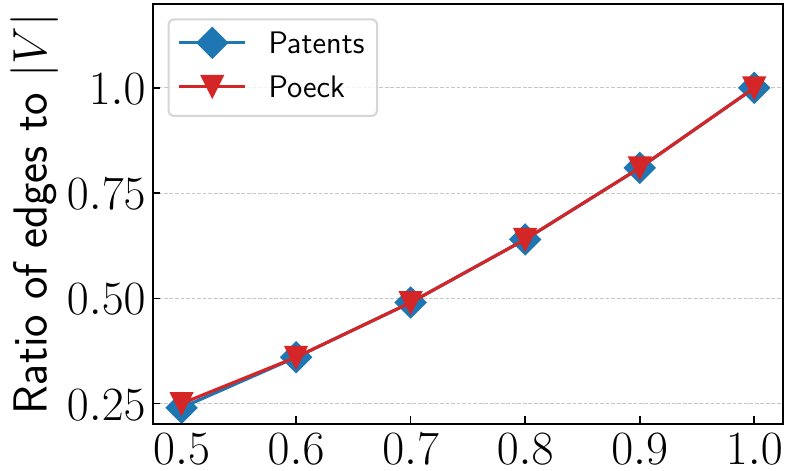}}
    \caption{Scalability evaluation by varying $|E|$ and $|V|$}
    \label{fig:sca}
\end{figure}

\myparagraph{Exp-6: Scalability evaluation by varying $|E|$ and $|V|$} In this experiment, we evaluate the scalability of \texttt{GAS} on two largest datasets, Patents and Pokec. The results are presented in Fig. \ref{fig:sca}.
To assess scalability, we randomly sample vertices and edges at rates between 50\% and 100\%, varying the number of vertices ($|V|$) and edges ($|E|$).
%To assess scalability, we varied the number of vertices ($|V|$) and edges ($|E|$) by randomly sampling vertices and edges at rates ranging from 50\% to 100\%.
For vertex sampling, we obtain subgraphs induced by the selected vertices. 
The runtime of \texttt{GAS} under different sampling rates is shown in Figs. \ref{fig:sca}(a) and \ref{fig:sca}(c), while Figs. \ref{fig:sca}(b) and \ref{fig:sca}(d) depict the vertex and edge ratios for the corresponding sampling scenarios.
The results demonstrate that the runtime of \texttt{GAS} scales smoothly as the number of vertices and edges grows.

%with the increase in the number of vertices and edges.

\begin{table}[t]
    \centering
    \footnotesize
    \caption{Upward route size evaluation}\label{table:route test}
    \renewcommand\arraystretch{1}
    \setlength{\tabcolsep}{1.1mm}
    \setlength{\extrarowheight}{2pt}
    %\resizebox{0.9\textwidth}{!}{
    \begin{tabular}{|l|c|c|c|c|}
    \hline
    \textbf{Datasets} &\textbf{Minimal size} & \textbf{Maximal size}&	\textbf{Sum size}&	\textbf{Average size}
    \\ \hline\hline
    College&	0&	60	& 32,314 &	2.34  \\ \hline
    Facebook&	0	&8,629	&1,478,230 &	14.55
    \\ \hline
    Brightkite	& 0	& 1,291	& 551,448 &	2.58
    \\ \hline
    Gowalla&	0&	633&	3,451,244	&  3.63
    \\ \hline
    Youtube	&0	& 1,555 &	5,533,322 &	1.85
    \\ \hline
    Google&	0	& 273	& 4,829,848	&1.12
    \\ \hline
    Patents	& 0 &	2,297	& 10,472,823	&0.63
    \\ \hline
    Pokec	& 0 &	971	& 64,276,694	&2.88
    \\ \hline
    \end{tabular} 
\end{table}

\begin{figure}[t]
    \centering
    \subfigure[Facebook]{\includegraphics[width=0.44\linewidth]{ 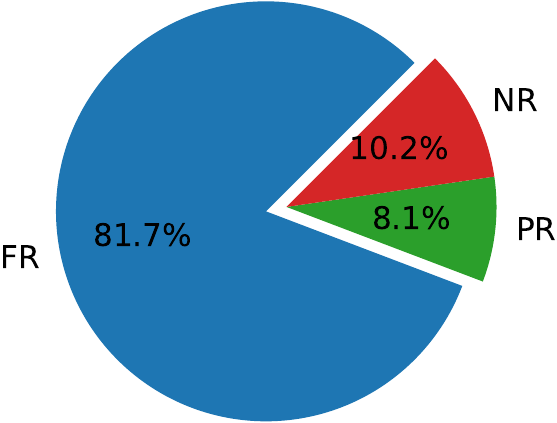}}
    \hspace{3mm}
    \subfigure[Gowalla]{\includegraphics[width=0.44\linewidth]{ 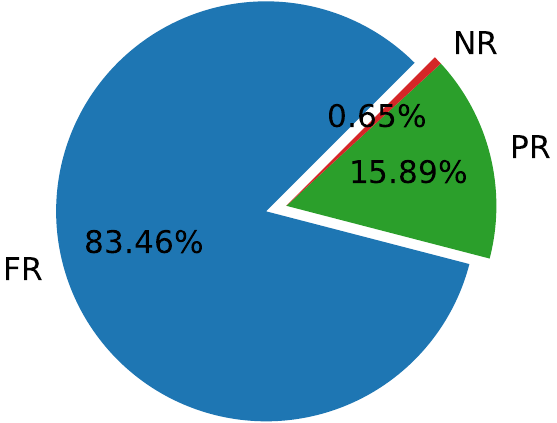}}
    \caption{Reuse test}
    \label{fig:RT}
\end{figure}

\myparagraph{Exp-7: Upward-route size evaluation} In this experiment, we evaluate the size of the upward-route for each edge during the first round of \texttt{GAS}.
Table \ref{table:route test} reports the route sizes for each dataset.
The results indicate that even the maximal route size constitutes only a small fraction of the entire graph.
The ``sum size'' refers to the total size of all upward-routes when each edge is considered as an anchor, which is at most 14 times the edge count ($|E|$) on Facebook. 
The "average size" is defined as the quotient of the sum size and $|E|$. 
With the upward-route optimization, only a limited number of edges need to be visited to compute the followers, thereby enabling the \texttt{BASE} algorithm to return results efficiently.

\myparagraph{Exp-8: Result reuse test} 
To evaluate the proportion of reusable results, we analyze the results computed in the first round of \texttt{GAS} that could be reused in subsequent rounds.
The reusable results are categorized into three groups: fully reusable (FR), partially reusable (PR), and non-reusable (NR).
Fully reusable results remain completely unchanged in the next round, while partially reusable results require recomputation only for the non-reusable tree nodes.
Non-reusable nodes, on the other hand, require complete re-computation.
As shown in Fig. \ref{fig:RT}, over 80\% of the results are fully reusable. This allows us to re-compute only the remaining results to identify the best anchor in the next round, significantly reducing the computation time.

\begin{table}[t]
    \centering
    \footnotesize
    \caption{Trussness gain, \texttt{AKT} v.s. \texttt{GAS}}\label{table:avg}
    \renewcommand\arraystretch{1.1}
    \setlength{\tabcolsep}{1.5mm}
    \setlength{\extrarowheight}{2pt}
    %\resizebox{0.9\textwidth}{!}{
    \begin{tabular}{|l|c|c|c|c|c|c|c|c|}
    \hline
    \textbf{Datasets} &\textbf{Col.} & \textbf{Fac.}&	\textbf{Bri.}&	\textbf{Gow.} &\textbf{You.} & \textbf{Goo.}&	\textbf{Pat.}&	\textbf{Pok.}
    \\ \hline\hline
    $avg_{gain}$&	51\% &	5\%	& 15\% & 20\%&	25\%& 27\%	& 25\% & 26\% \\ \hline
    $max_{gain}$&	74\% &  8\%	& 23\% & 31\% &	42\%& 35\%	& 47\% & 47\%
    \\ \hline
    \end{tabular} 
\end{table}

\begin{figure}[t]
    \centering
    \subfigure[\texttt{AKT}]{\includegraphics[width=0.45\linewidth]{ 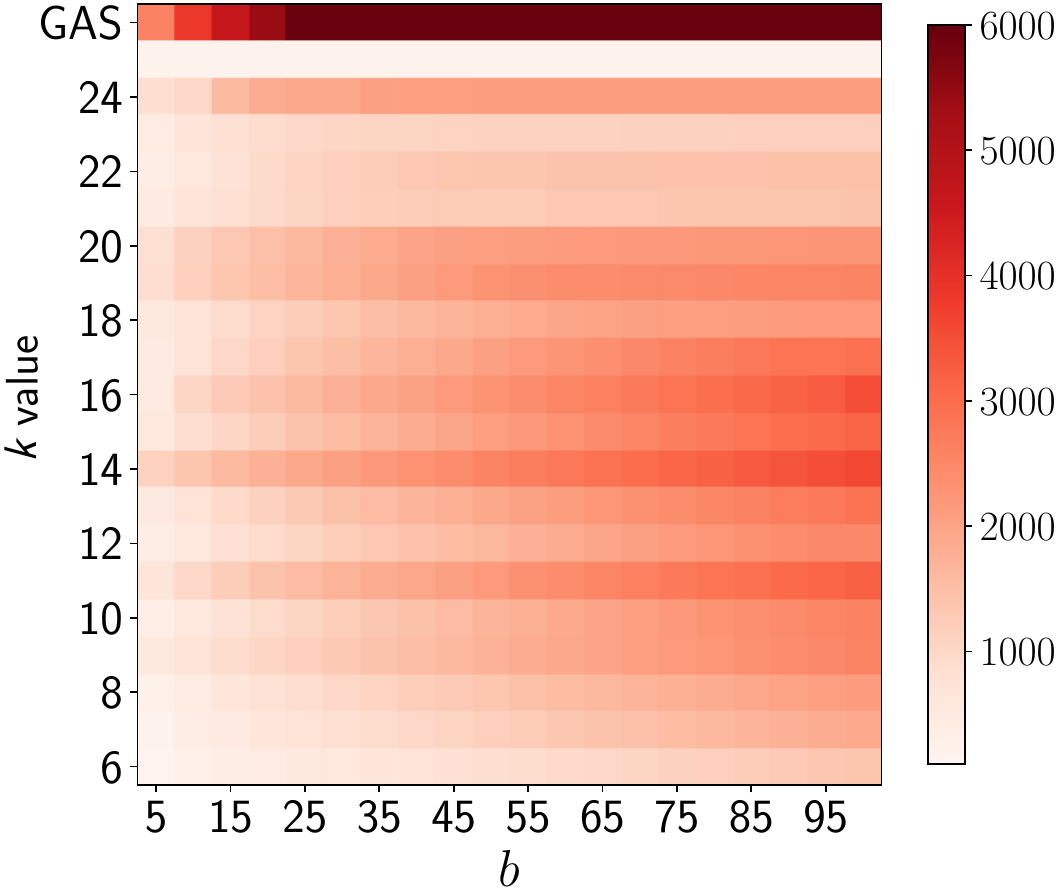}}
    \hspace{3mm}
    \subfigure[\texttt{GAS}]{\includegraphics[width=0.45\linewidth]
    { 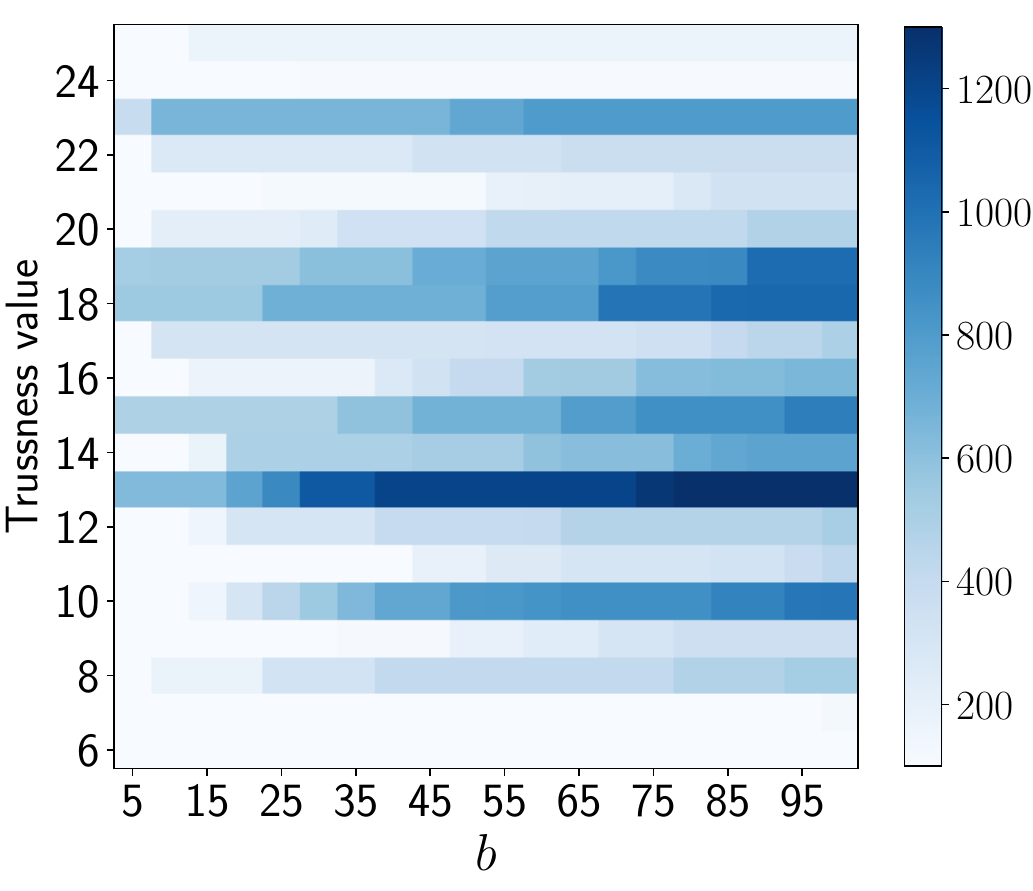}}
    \caption{Trussness gain distribution on Gowalla}
    \label{fig:dt}
\end{figure}

\myparagraph{Exp-9: Comparison with \texttt{AKT}}
%\textcolor{blue}{
In this experiment, we conduct a detailed comparison between \texttt{GAS} and \texttt{AKT} \cite{zhangfanefficiently2018}. TABLE \ref{table:avg} presents the trussness gain ratio of \texttt{AKT} to \texttt{GAS} when $b=50$ on all datasets. $max_{gain}$ represents the maximum trussness gain achieved by \texttt{AKT} for all possible $k$, while $avg_{gain}$ denotes the average trussness gain over all $k$ values. The results indicate that even at the optimal $k$ value, \texttt{AKT} achieves only 8\% to 72\% of the trussness gain obtained by \texttt{GAS}.
Furthermore, Fig. \ref{fig:dt}(a) illustrates a detailed comparison between \texttt{AKT} and \texttt{GAS} on dataset Gowalla. In the heatmap, each grid represents the trussness gain achieved by \texttt{AKT} for a given $k$ and $b$, with different color intensities indicating varying levels of gain. For better comparison, we overlay the trussness gain of \texttt{GAS} at the top, showing its performance under different budget values. The results clearly demonstrate that \texttt{GAS} consistently outperforms \texttt{AKT} across all parameter settings, achieving significantly higher trussness gains.
Additionally, Fig. \ref{fig:dt}(b) visualizes the distribution of \texttt{GAS}'s followers across different trussness levels. Each grid in the heatmap represents the number of followers at a given trussness level for a specific budget. The figure reveals that \texttt{GAS}'s followers span a wide range of trussness values, highlighting the advantage of our global optimization strategy, which effectively enhances trussness across the entire graph rather than being constrained to a specific $k$-truss subgraph.

\section{Related works}

\label{sec:rel}

%intruoduce subgraph models
There are many cohesive models studied in different scenarios, such as clique \cite{coenfindingclique,cheng2011finding,sun2020discovering,sun2023clique,chen2021maximum}, quasi-clique \cite{gao2016detecting}, $k$-core \cite{franDistance,chen2021edge,sun2020stable,HongboCritical}, $k$-ECC \cite{Lijun-k-edge,jiafengQuerying}, $k$-truss \cite{cohen2008trusses,jiawangTrussdecomposition}.
Among these models, $k$-truss and $k$-core are two widely used models, since both $k$-truss and $k$-core can be computed within polynomial time complexity  \cite{jiawangTrussdecomposition,franDistance,Vladimir2003An}. 
%introduce some applications
%introduce some works by using these models
%truss
Cohen \etal \cite{cohen2008trusses} introduced the $k$-truss model, which defines the maximal subgraph where each edge is contained in at least $k-2$ triangles.
%Cohen \etal \cite{cohen2008trusses} proposed the $k$-truss model which is the maximal subgraph where each edge must be contained in at least $k-2$ triangles.
The $k$-truss model has various applications, such as community search
%The $k$-truss model has a wide range of applications like community search 
\cite{huang2014querying,alifirmtruss2022,VLDB-eff-community-search,Esratruss2017,liuqingtrussbesed,xiexiaoqinefficient}, P2P networks \cite{zitansunadaptive}, Urban and transportation Networks \cite{BuOnimproving}. 
Several algorithms were proposed to efficiently compute the $k$-truss of a graph. The most notable early work by Cohen \cite{cohen2008trusses} proposed a straightforward algorithm based on triangle enumeration. 
Wang \etal \cite{jiawangTrussdecomposition} proposed a efficient $O(m^{1.5})$ algorithm to solve the problem.
Behrouz \etal \cite{alifirmtruss2022} extended $k$-truss to multilayer graphs.
Furthermore, $k$-truss is often used as a powerful tool to measure user engagement or relationship importance. 
Thus, the truss maximization and minimization problem was widely studied, and such problems can be done in different ways.
Zhang \etal \cite{zhangfanefficiently2018} introduced the \texttt{AKT} algorithm by vertices anchoring to maximize $k$-truss vertices size, enhancing user engagement and tie Strength. 
Bu \etal \cite{BuOnimproving} developed an efficient algorithm to maximize $k$-truss by wisely choosing vertices to merge.  
And Sun \etal \cite{XinSunbudget2021,zitansunadaptive}
proposed a component-based algorithm and minimum approach to efficiently enlarge $k$-truss by adding edges. 
Zhu \etal \cite{DBLP:conf/ijcai/ZhuZCWZL19} tried to minimize $k$-truss by wisely choosing deleting edges.
What's more, the concept of $k$-truss has been extended to various graph types, including directed graphs \cite{liuqingtrussbesed}, weighted graphs \cite{zhengzibinweightedtruss}, signed graphs \cite{wu2020maximum,zhaojunCommunity}, multilayer graphs \cite{alifirmtruss2022,huangtrussdecommultillayer}. 
Due to the popularity of the $k$-truss, truss maintenance has been strongly required. 
Several studies \cite{huang2014querying,efficientSunzitian,QiluoExploring2023,Qiluobatch2020,ikaiUn2019} proposed some efficient maintenance algorithms to efficiently update trussness when the graph changes.
%The high density of $k$-truss has led to the use of $k$-truss in many community search studies\cite{alifirmtruss2022,huang2014querying,VLDB-eff-community-search,Esratruss2017,liuqingtrussbesed,xiexiaoqinefficient}. 

 %Thus a lot of work focus on truss maintenance \cite{huang2014querying,efficientSunzitian,QiluoExploring2023,Qiluobatch2020,yikaiFast2017}.

%core

%Besides, proposed by Seidman \cite{seidman1983network}, the $k$-core, which is a maximal subgraph where each vertex has at least $k$ neighbors, is also a widely used model to identify critical vertices or relationships.
%Some classic problem like AK \cite{kshpreventing2015,zhang2017olak} and AC \cite{QingyuanGobal2020,Tengsiyioptimizing} problem which aim to anchor $b$ vertices to enhance the user engagement. 
%Similarly to the $k$-truss, core minimization and maximization problems are also popular. 
%\cite{XinSunfast2022} proposed a novel strategy to insert edges to enlarge corresponding $k$-core while \cite{HongboCritical} tried to enlarge $k$-core by merging vertices.
%Zhu \etal \cite{weijiezhukcoremax} and Zhang \etal \cite{zhang2017finding} tried to minimize $k$-core by deleting edges and removing vertices, respectively. 
%There are also works focused on core maintenance on dynamic graphs to efficiently update coreness \cite{li2013efficient,ikaiUn2019,zhangyikaiMaintaning}.

\section{Conclusion}
\label{sec:conc}
%In this paper, we introduce the Anchor Trussness Reinforcement Problem (ATP), which aims to select $b$ edges as anchors to maximize trussness gain.
In this paper, we present the ATR problem, which seeks to select a set of $b$ edges from a graph as anchor edges in order to maximize the overall trussness gain of the network.
%which aims to select a set of $b$ edges from a graph as anchor edges to maximize the overall trussness gain of the network.
We also prove that this problem is NP-hard.
To address this challenge, we propose a greedy framework. 
To efficiently narrow down the search space of followers, we introduce the concept of upward routes and proved that only edges along the upward route may increase their trussness. Combined with a support check process, this approach enables fast and precise identification of followers. 
Additionally, to avoid redundant computations of unchanged results from previous rounds, we propose a followers classification tree, which effectively classifies followers. After selecting an anchor, the algorithm only processes tree nodes with structural changes and reuses previously computed results wherever applicable.
Finally, we performe comprehensive experiments on 8 datasets to assess the effectiveness and efficiency of our approach.

\noindent \textbf{Acknowledgments}.
This work was supported by UoW R6288 and ARC DP240101322, DP230101445.
\bibliographystyle{IEEEtran}  %xxx

% \balance
\bibliography{paper}

\clearpage

\end{document}